\documentclass[final,3p,times]{elsarticle}
\usepackage{amsmath}    % need for subequations
\usepackage{graphicx}   % need for figures
\usepackage{array}      % use for table
\usepackage{booktabs}   % use for tabs
\usepackage{verbatim}   % useful for program listings
\usepackage{color}      % use if color is used in text
\usepackage{subfigure}  % use for side-by-side figures
\usepackage{hyperref}   % use for hypertext links, including those to external documents and URLs
\usepackage{fancyheadings}
\usepackage{mathtools}
\usepackage{amsthm}
\usepackage{amsxtra}
\usepackage{amstext}
\usepackage{amssymb}
\usepackage{enumerate}
\usepackage{graphicx}
\usepackage{rotating}   % use for producing landscape tables
\usepackage[T1]{fontenc}
\usepackage[utf8]{inputenc}
\usepackage[font=small,labelfont=bf,tableposition=top]{caption}
\usepackage{booktabs,multirow}
\begin{document}

\begin{frontmatter}
\title{Hilfer fractional advection-diffusion equations with
power-law initial condition; a Numerical study using variational
iteration method}

\author{\textbf{Iftikhar Ali$^1$ and Nadeem Malik}}

\address{
Department of Mathematics and Statistics\\
King Fahd University of Petroleum and Minerals\\
Box 5046 Dhahran, 31261, Saudi Arabia.\\
{\large\rm \[Under\ Review,\ Computers\ And\ Mathematics\ With\ Applications\]}}

\fntext[label2]
{Corresponding author.\\
E-mail addresses: iali@kfupm.edu.sa (I. Ali), namalik@kfupm.edu.sa and nadeem\_malik@cantab.net (N. Malik).}

\begin{abstract}
We propose a Hilfer advection-diffusion equation of order
$0<\alpha<1$ and type $0\leq\beta\leq1$, and find the power series
solution by using variational iteration method. Power series
solutions are expressed in a form that is easy to implement
numerically and in some particular cases, solutions are expressed
in terms of  Mittag-Leffler function. Absolute convergence of
power series solutions is proved and the sensitivity of the
solutions is discussed with respect to changes in the values of
different parameters. For power law initial conditions it is
shown that the Hilfer advection-diffusion PDE gives the same
solutions as the Caputo and Riemann-Liouville advection-diffusion
PDE. To leading order, the fractional solution compared to the non-fractional solution
increases rapidly with $\alpha$ for $\alpha > 0.7$ at a given time
$t$; but for $\alpha<0.7$ this factor is weakly sensitive to
$\alpha$. We also show that the truncation errors, arising when
using the partial sum as approximate solutions, decay
exponentially fast with the number of terms $n$ used. We find that
for $\alpha< 0.7$ the number of terms needed is weakly sensitive
  to the accuracy level and to the fractional order, $n\approx 20$;
   but for $\alpha>0.7$ the required number of terms increases rapidly
    with the accuracy level and also with the fractional order $\alpha$.
\end{abstract}

%\maketitle \numberwithin{equation}{section}
\newtheorem{theorem}{Theorem}[section]
\newtheorem{lemma}[theorem]{Lemma}
\newtheorem{remark}[theorem]{Remark}
\newtheorem{definition}[theorem]{Definition}
\newtheorem{corollary}[theorem]{Corollary}
\allowdisplaybreaks


\end{frontmatter}

\section{Introduction}\label{intro}

Many transport phenomenon such as the time evolution of chemical
or biological species in a flow field are often modeled by partial
differential equations. These PDE's are of the
advection-diffusion-reaction type and can be derived from mass
balance, momentum balance and energy balance equations and in the
case of multi-component species, we also have individual species
mass balance equations, see \cite{hund} and \cite{bird}.

In a single-phase single-component system, let us denote the
scalar concentration field by $c(x,t)$ at the position $x$ and at
the time instant $t$, whose transport is described by the non-linear
advection-diffusion-reaction equation,
\begin{equation}\label{eq7}
\frac{\partial{c(x,t)}}{\partial{t}}+\frac{\partial}{\partial
x}P(u(x,t),c(x,t))=-\frac{\partial}{\partial
x}\left(J_c\right)+f(x,t,c(x,t)).
\end{equation}
where $u(x,t)$ is the velocity field in which the scalar is transported, $P(u,c)$ is a non-linear convective flux, $J_c(x,t)$ is the scalar flux, and $f(x,t,c(x,t))$ is the source/reaction term. See Appendix A for details on how this eqaution is derived.

Many physical phenomena appearing in the studies of fluid
mechanics, astrophysics, ground water flow, \cite{saenton},
meteorology, \cite{marchuk}, \cite{holton}, semiconductors,
\cite{ellis},  and reactive flows, \cite{oran}, are modeled by Eq.
\eqref{eq7}. The nonlinear advection-diffusion equation also
proves to be effective in describing the behavior of two-phase
flow in oil reservoir, \cite{aziz}, non-newtonian flows,
\cite{glow},  front propagation, \cite{menko}, traffic flow,
\cite{kest}, financial modeling, \cite{piro}.

Although the above mathematical models adequately describe a lot of natural
phenomena, there still exist many complex phenomena in nature
which are not described adequately by these models. Among them are
crowded systems, such as protein diffusion within cells,
\cite{weiss}, and diffusion through porous media, \cite{chen}.
 So there is a need to develop new models to understand
 such complex phenomena.  In this regard fractional
calculus could be helpful in describing such complex phenomena,
see \cite{old}, \cite{mlros}, \cite{igor} and \cite{dieth}. For 
instance, in order to obtain a better understanding of anomalous
diffusion Caputo \cite{cap99} used fractional calculus to incorporates
memory.

Fractional calculus continues to attract the attention of researchers
in physics, biology, chemistry and other engineering sciences \cite{saba}. 
The reason lies in its ability to explain the complex systems,
such as, anomalous diffusion in porous media, crowding in living cells, time evolutionary
processes which depend on the past history. Experimental evidence of anomalous diffusion have
been reported by Hilfer \cite{hilfer2} while working on dielectric spectroscopy
(in particular, glassy formations and relaxations in polymers). Joen et al.
\cite{jeon} have reported their findings  about the evidence of anomalous diffusion
in living organism, and Tabie et al. \cite{tabei}, during the study of crowded
systems (intra cellular transport of insulin in granules), have witnessed the anomalous
diffusion of insulin in cells.

Anomalous diffusion, and transport through porous medium, can be
understood as a random walk processes, \cite{havlin},  especially
continuous time random walk (CTRW) models. Conventional, Brownian
CTRW, is characterised by waiting times and jumps in particle
location whose probability density function are Gaussian and the
pdf obeys the classical advection-diffusion equation. Anomalous
diffusion, on the other hand, possesses pdf's of waiting times and
jumps which are inverse power laws, and it can be shown that such
a process is described by fractional advection-diffusion
equations.

Furthermore, in a standard diffusion process the mean square
displacement, MSD, $<x(t)^2>$  of a particle is linearly related
to the time $t$, but in anomalous diffusion MSD $<x(t)^2>$ has a
nonlinear power relationship with the time $t$, that is, $<x(t)^2>
\propto t^{\alpha}$. For $0<\alpha<1$ diffusion is called
subdiffusion; for $\alpha=1$, we have the standard diffusion; for
$1<\alpha<2$ the diffusion is termed as superdiffusion. Mean
square displacement can be understood geometrically as the amount
of space the particle has explored in the system. For more
details, see Metzler and Klafter \cite{metzler1}. They have
obtained the following fractional advection-diffusion equation,
\begin{equation}\label{eq10} %{eq12}
\frac{\partial}{\partial
t}c(x,t) =\ _0\mathcal{D}^{1-\alpha}_t\left[K_{\alpha}\frac{\partial^2
c(x,t)}{\partial x^2}-A_{\alpha}\frac{\partial}{\partial
x}\left\{u(x)c(x,t)\right\}\right].
\end{equation}
where $c(x,t)$ is a scalar field,
$_0\mathcal{D}^{1-\alpha}_t$ is the Riemann-Liouville
fractional derivative of order $0<\alpha<1$ defined in Eq.
\eqref{pre2}, $K_{\alpha}$ and $A_{\alpha}$ are called the
generalized diffusion constants.  % defined in \cite{metzler1}.

 Note that Eq. \eqref{eq10} reduces to the problem considered by Caputo \cite{cap99}
 by taking $A_{\alpha}=0$ and replacing $1-\alpha$ by $\alpha$
 (using the same symbol). Also note that by setting
 $K_{\alpha}=1$,  $A_{\alpha}=1$, and $\alpha=1/2$ in Eq.
 \eqref{eq10}, we obtain a time fractional diffusion equation of
 order $1/2$ which was considered by Das in \cite{sdas} and
 by Saha in \cite{saha}.

In the present study, we propose a similar equation to \eqref{eq10}
but with the Riemann-Liouville fractional derivative
$_0\mathcal{D}^{1-\alpha}_t$ is replaced by the Hilfer fractional
derivative $_0\mathcal{D}_t^{\alpha,\beta}$, (defined later in Eq.
\eqref{pre5}). Hilfer fractional derivative
$_0\mathcal{D}_t^{\alpha,\beta}$ is a sort of interpolation between
the Riemann-Liouville fractional derivative and Caputo fractional
derivative, see \cite{hilfer1}. Example of a physical system that
can be modeled by Hilfer fractional derivative is given by Hilfer
in \cite{hilfer2}.

\section{Problem statement: Hilfer fractional advection-diffusion system}

We interpret $c(x,t)$ as a diffusive scalar field (for example temperature or
concentration), and we use the simplifying assumption
$K_{\alpha}=A_{\alpha}=\kappa=$constant in Eq. \eqref{eq10} and
obtain the following linear fractional advection-diffusion
equation,
\begin{equation}\label{eq11}%{e1}
\frac{\partial{c(x,t)}}{\partial{t}} = \kappa\
_0\mathcal{D}_t^{\alpha,\beta}\left[\frac{\partial^2c(x,t)}{\partial
x^2}-\frac{\partial}{\partial x}\left\{u(x)c(x,t)\right\}\right],\quad x
> 0, t > 0.
\end{equation}

Equation \eqref{eq11} is called the Hilfer advection-diffusion
equation of order $\alpha$ and type $\beta$; $\kappa>0$  represents diffusivity; $u(x)$ represents the velocity
field. Note that equation \eqref{eq11} reduces to equation \eqref{eq10}
by taking $\beta = 0$ and relabeling $\alpha$ by $1-\alpha$. Also note that
in \cite{metzler2} Sandev considers a diffusion-reaction equation that
involves Hilfer fractional derivative but without the advection term.

Here we outline the main objectives of the present study. Firstly,
we find the power series solution of Eq. \eqref{eq11}, with
$u(x)=-x$ and initial condition $c(x,0)=f(x)$ by using variational
iteration method, described in Section (\ref{secvim}). Secondly,
we represent the power series solution in a convenient form that
is easy to use for numerical purposes, especially we give a
recurrence relation for the $x-$part in the $n$th term of the
series. Thirdly, we prove the absolute convergence of the series
solution accompanied by some examples. Fourthly, we analyze the
behavior of the fractional solution with respect to the parameters
$\alpha$, $\kappa$, and $p$; and we discuss the numerical
convergence of the solutions which arise when we use the truncated
series solution. Finally, we examine the fractional solutions for
$\alpha>0$ and compare with the conventional solution for
$\alpha=0$ in order to elucidate trends in the solution as
$\alpha$ increases for different parameter values $p$ and
$\kappa$.

 We have organized this paper as follows: in Section \ref{prelim},
  we provide some basic definitions and results from
 fractional calculus; in Section \ref{secvim}, we describe
the variational iteration method to obtain  the solution of the
problem \eqref{eq11} subject to initial condition; in Section
\ref{case}, we present a case study of polynomial uploading; in
Section \ref{num}, we discuss the numerical results and provide
graphs of the solutions along with error analysis; in Section
\ref{FrVsCon}, we compare the fractional solutions with the
corresponding conventional versions; and in the last Section
\ref{conc}, we state our conclusions of the study.
\section{Preliminaries}\label{prelim}
In this section, we briefly discuss the importance and significance
of time fractional derivatives and Hilfer-composite
time fractional derivative, and state some definitions and results from
fractional calculus, see \cite{metzler2, kilbas}.

It was shown by Hilfer that time fractional derivatives are equivalent to
infinitesimal generators of generalized time fractional evolutions, which
arise in the transition from microscopic to macroscopic time scales \cite{hilfer2,hilfer3}.
Hilfer showed that this transition from ordinary time derivative to fractional
time derivative indeed arises in physical problems \cite{hilfer1,hilfer4,metzler2,tom}.

%Z. Tomovski, T. Sandev, R. Metzler and J. Dubbeldam, Physica A 391 2527 (2012)].

\vskip 0.2cm \noindent\textbf{\emph{Riemann-Liouville Fractional
Integral }}
 of order $\alpha$ for an absolutely integrable function $f(t)$ is
defined by
\begin{equation}\label{pre1} %12
\left(_0I_t^{\alpha}f\right)(t):=\frac{1}{\Gamma(\alpha)}\int_0^t \frac{f(\tau)}{(t-\tau)^{1-\alpha}} d\tau,\quad
t > 0, \alpha>0
\end{equation}
when the right hand side exists.
\vskip 0.2cm \noindent\textbf{\emph{Riemann-Liouville Fractional
Derivative }}%\vskip 0.2cm
 of
order $\alpha>0$ for an absolutely integrable function $f(t)$ is
defined by
\begin{equation}\label{pre2} %13
\left(_0D_t^{\alpha}f\right)(t)=\frac{1}{\Gamma(1-\alpha)}\frac{d}{dt}\int_{0}^{t}\frac{f(\tau)}{(t-\tau)^{\alpha}}d\tau,
\qquad t>0, \quad 0<\alpha<1
\end{equation}
\vskip 0.2cm \noindent\textbf{\emph{Caputo Fractional
Derivative }}%\vskip 0.2cm
 of
order $\alpha>0$ for a function $f(t)$, whose first derivative is
absolutely integrable, is defined by

\begin{equation}\label{pre3} %14
\left(_0^*
D_t^{\alpha}f\right)(t)=\frac{1}{\Gamma(1-\alpha)}\int_{0}^{t}\frac{f^{'}(\tau)}{(t-\tau)^{\alpha}}d\tau,
 \qquad t>0, \quad 0<\alpha<1
\end{equation}

\vskip 0.2cm \noindent\textbf{\emph{Relationship between
Riemann-Liouville and Caputo Fractional Derivative}} \vskip 0.1cm

\begin{gather}\label{pre4} %15
\begin{aligned}
\left(_0^*D_t^{\alpha}f\right)(t)& =\ _0D_t^{\alpha}f(t)-f(0^+)\frac{t^{-\alpha}}{\Gamma(1-\alpha)}\\
&=\ _0D_t^{\alpha}\left[f(t)-f(0^+)\right].
\end{aligned}
\end{gather}

\vskip 0.2cm \noindent\textbf{\emph{Hilfer Fractional
Derivative }}%\vskip 0.2cm
 of
order $\alpha>0$ and type $\beta$ for an absolutely integrable
function $f(t)$ with respect to $t$ is defined by,
\begin{equation}\label{pre5} %16
\left( _0D_t^{\alpha,\beta}f\right)(t)=\left(_0I_t^{\beta(1-\alpha)}\frac{d}{dt}\ _0I_t^{(1-\beta)(1-\alpha)}f\right)(t),
\qquad t>0, \quad 0<\alpha<1, 0\leq\beta\leq1.
\end{equation}

\begin{lemma}\label{lma1}\cite{srivastava2009fractional}
  The following fractional derivative formula holds true:
\begin{equation}\label{pre9}
_0D_t^{\alpha,\beta}\left( t^{\gamma}\right) = \frac{\Gamma(1+\gamma)}{\Gamma(1+\gamma-\alpha)}\left( t^{\gamma - \alpha}\right),
\qquad t>0, \quad \gamma > -1,
\end{equation}
where $0 < \alpha < 1$ and $0 \leq \beta  < 1$.
\end{lemma}
From lemma \ref{lma1}, We can easily have the following lemma.

\begin{lemma}\label{lma2}
  On integrating equation \eqref{pre9}, we obtain
\begin{equation}\label{pre7}
\int_0^t \ _0D_t^{\alpha,\beta}\left( t^{\gamma}\right) dt = \frac{\Gamma(1+\gamma)}{\Gamma(1-\alpha+\gamma+1)}
\left( t^{1- \alpha + \gamma}\right), \qquad t>0, \quad \gamma > -1,
\end{equation}
where $0 < \alpha < 1$ and $0 \leq \beta < 1$.
\end{lemma}

\noindent\textbf{\emph{Remarks:}} \vskip 0.2cm
\begin{enumerate}
\item The Caputo derivative represents a type of regularization in the
time domain (origin) for Riemann-Liouville derivative.
\item Hilfer fractional derivative interpolates between
Riemann-Liouville fractional derivative and Caputo fractional
derivative, because if $\beta=0$ then Hilfer fractional derivative
corresponds to Riemann-Liouville fractional derivative and if
$\beta=1$ then Hilfer fractional derivative corresponds to Caputo
fractional derivative.
\item $f(0^+)$ is required to be finite.
\item The three derivatives are equal if $f$ is continuous on $[0,T]$ and $f(0^+)=0$,
see Lemma \ref{lma3}.
\end{enumerate}

\vskip 0.2cm \noindent\textbf{\emph{Mittag-Leffler Function }}
\vskip 0.2cm Mittag-Leffler function is the generalization of
exponential function $e^z=\sum_{k=0}^{\infty}\frac{z^k}{k!}$.
\vskip 0.2cm \noindent\textbf{\emph{1-parameter Mittag-Leffler
Function }}

\begin{equation}\label{pre7}%{pre8}
E_{\alpha}(z)=\sum_{k=0}^{\infty}\frac{z^k}{\Gamma(\alpha k+1)},
\quad \alpha>0.
\end{equation}

\vskip 0.2cm \noindent\textbf{\emph{2-parameter Mittag-Leffler
Function }}
\begin{equation}\label{pre8}%{pre9}
E_{\alpha,\beta}(z)=\sum_{k=0}^{\infty}\frac{z^k}{\Gamma(\alpha
k+\beta)}, \quad \alpha>0, \beta>0.
\end{equation}

The motivation of studying fractional equations of form \eqref{eq11} is,
from one side, the Hilfer generalized time fractional derivative \eqref{pre5},
which combine both the derivatives, Caputo and R–L. It is known, from the
continuous time random walk (CTRW) theory, that the probability density
$f(x, t)$, in case where the characteristic waiting time diverges and the
jump length variance is finite, can be obtained from the following two
equivalent representations of the fractional diffusion equation \cite{metzler3, tom}
$$
_0D_{t}^\mu f(x,t) -f(x,0^{+}) \frac{t^{-\mu}}{\Gamma(1-\mu)}
= \kappa \frac{\partial^2}{\partial x^2}f(x,t)
$$
$$
_0^*D_{t+}^\mu f(x,t) = \kappa \frac{\partial^2}{\partial x^2}f(x,t)
$$
in the R–L and Caputo sense, respectively, where $\kappa$ is the generalized
diffusion constant of physical dimension $[\kappa] = m^2/s^\mu$, and $\mu$
is the anomalous diffusion exponent. Thus, if the initial conditions are
properly taken into account, the Caputo and Riemann-liouville formulations
of the time-fractional diffusion-advection equations  are identical.

\section{Variational Iteration Method}\label{secvim}
In this section, we describe the variational iteration method,
\cite{waz}, and provide an outline for its implementation. The VIM
has been used extensively by several authors, see \cite{zhao,
noor, obidat}, in recent years to obtain series solutions of
problems arising in different areas of applied mathematics and
engineering. VIM has been successfully applied to solve problems
like Riccati equation, heat equation, wave equation and many other
problems. Ji-Huan He, \cite{he1, he2}, proposed VIM to obtain the
solutions of nonlinear differential equations. The method provides
the solution in the form of a successive approximations that may
converge to the exact solution if such a solution exists. In case
where a closed form of the exact solution is not achievable, we
use the truncated series, for instance, the $n$th partial sum of
the series. VIM has certain advantages over the other proposed
methods like Adomian decomposition method (ADM), \cite{waz}, and
homotopy perturbation method (HPM), see \cite{obidat}. In the case
of ADM a lot of work has to be done to compute the Adomian
polynomials for nonlinear terms and in the case of HPM, a huge
amount of calculation has to be done when degree of nonlinearity
increases. On the other hand,
 no specific requirements are needed for nonlinear operators in
order to use VIM. For instance, HPM requires an introduction of
small parameter that is sometimes difficult to incorporate in the
equation or its introduction may change the physics of the
problem. \vskip 0.2cm \par The basic concepts and main steps for
the implementation of VIM are explained here. Consider the
following equation:
\begin{equation}\label{vim1}
\frac{\partial}{\partial
t}c(x,t) =\ _0\mathcal{D}_t^{\alpha,\beta}[(Ac)(x,t)]+g(x,t),
\end{equation}
where $_0\mathcal{D}_t^{\alpha,\beta}(.)$ represents the Hilfer
fractional derivative with respect to time variable $t$, and $A$
represents a differential operator with respect to variable $x$.
\vskip 0.2cm \par The variational iteration method presents a
correctional functional in $t$-direction for Eq. \eqref{vim1} in
the form,
\begin{equation}\label{vim2}
c_{n+1}(x,t)=c_{n}(x,t)+\int_0^t\lambda(\xi)\left(\frac{\partial{c_n(x,\xi)}}{\partial{\xi}} - \
_0\mathcal{D}_t^{\alpha,\beta}[A(\widetilde{c_n}(x,t))]-g(x,t)\right)d\xi,
\end{equation}
with $c_n$ assumed known, where $\lambda(\xi)$ is a general
Lagrange multiplier which can be identified optimally by
variational theory and $\tilde{c}_n$ is a restricted value that
means it behaves like a constant, hence $\delta\widetilde{c}_n=0$,
where $\delta$  is the variational derivative. \vskip 0.2cm \par
The VIM is implemented in two basic steps, see \cite{waz};
\begin{enumerate}
\item the determination of the Lagrange multiplier $\lambda(\xi)$
that will be identified optimally through variational theory, and
\item with $\lambda(\xi)$ determined, we substitute the result
into Eq. \eqref{vim2} where the restriction should be omitted.
\end{enumerate}
\vskip 0.2cm \par Taking the $\delta-$variation of Eq.
\eqref{vim2} with respect to $c_n$, we obtain
\begin{equation}\label{vim3}
\delta c_{n+1}(x,t)=\delta
c_{n}(x,t)+\delta\int_0^t\lambda(\xi)\left(\frac{\partial{c_n(x,\xi)}}{\partial{\xi}} - \
_0\mathcal{D}_t^{\alpha,\beta}[A(\widetilde{c_n}(x,t))]-g(x,t)\right)d\xi.
\end{equation}
\vskip 0.2cm \par Since $\delta \widetilde{c_n}=0$ and $\delta
g=0$, we have
\begin{equation}\label{vim4}
\delta c_{n+1}(x,t)=\delta
c_{n}(x,t)+\delta\int_0^t\lambda(\xi)\left(\frac{\partial{c_n(x,\xi)}}{\partial{\xi}}\right)d\xi.
\end{equation}
\vskip 0.2cm \par To determine the Lagrange multiplier
$\lambda(\xi)$,  we integrate by parts the integral in  Eq.
\eqref{vim4}, and noting that variational derivative of a constant
is zero, that is, $\delta k=0$. Hence  Eq. \eqref{vim4} yields
\begin{align}\label{vim5}
\nonumber \delta c_{n+1}(x,t)&=\delta c_{n}(x,t)+\delta
c_n(x,\xi)\lambda(\xi)|_{\xi=t}-\int_0^t\frac{\partial}{\partial\xi}\lambda(\xi)\delta
c_n(x,\xi)d\xi\\
&=\delta
c_{n}(x,t)(1+\lambda(\xi)|_{\xi=t})-\int_0^t\frac{\partial}{\partial\xi}\lambda(\xi)\delta
c_n(x,\xi)d\xi
\end{align}
\vskip 0.2cm \par The extreme values of $c_{n+1}$ requires that
$\delta c_{n+1}=0$. This means that left hand side of Eq.
\eqref{vim5} is zero, and as a result the right hand side should
be zero as well, that is,
\begin{equation}\label{vim6}
\delta
c_{n}(x,t)(1+\lambda(\xi)|_{\xi=t})-\int_0^t\frac{\partial}{\partial\xi}\lambda(\xi)\delta
c_n(x,\xi)d\xi=0
\end{equation}
\vskip 0.2cm \par This yields the stationary conditions
\begin{gather}
1+\lambda(\xi)|_{\xi=t}=0\\
\text{and }\lambda'(\xi)=0\\
\text{which implies } \lambda=-1.
\end{gather}
\vskip 0.2cm \par Hence  Eq. \eqref{vim2} becomes
\begin{equation}\label{vim7}
c_{n+1}(x,t)=c_{n}(x,t)-\int_0^t\left(\frac{\partial{c_n(x,\xi)}}{\partial{\xi}} - \
_0\mathcal{D}_t^{\alpha,\beta}[A(c_n(x,t))]-g(x,t)\right)d\xi,
\end{equation}
where the restriction is removed on $c_n$. We can use Eq.
\eqref{vim7} to obtain the successive approximation of the
solution of the problem \eqref{vim1}. The zeroth approximation
$c_0(x,t)$ can be chosen in such away that it satisfies  the
initial condition and the boundary conditions. Appropriate
selection of the zeroth approximation is necessary for the
convergence of the successive approximation to the exact solution
of the problem.

However, we remark that beacause VIM involves derivatives of all
$c_n$'s inside the integral in the above equation, then VIM is limited to smooth initial conditions. It is not suitable for initial conditions such as $c_0(x)=\delta(x)$ -- the latter would produce a Green's function for the
physical problem. Nevertheless, provided smooth initial conditions
can be specified then VIM is often fast and very effective, as demonstrated
in the case studies below.

\subsection{Solution of the Problem} We consider the equation

\begin{equation}\label{vim8}  %{e2}
\frac{\partial{c(x,t)}}{\partial{t}}=\kappa
\ _0\mathcal{D}_t^{\alpha,\beta}\left[\frac{\partial^2c(x,t)}{\partial
x^2}-\frac{\partial}{\partial
x}\left\{u(x)c(x,t)\right\}\right],\quad x > 0, t > 0
\end{equation}
with initial condition $c(x,0)=f(x).$ \vskip 0.2cm \par According
to the variational iteration method, we consider the correctional
functional in $t$-direction by using Eq. \eqref{vim2}
\begin{equation}\label{vim9} %{e3}
c_{n+1}(x,t)=c_{n}(x,t)+\int_0^t\lambda(\xi)\left(\frac{\partial{c_n(x,\xi)}}{\partial{\xi}}-\kappa
\ _0\mathcal{D}_\xi^{\alpha,\beta}\left[\frac{\partial^2\widetilde{c}_n(x,\xi)}{\partial
x^2}-\frac{\partial}{\partial
x}\left\{u(x)\widetilde{c}_n(x,\xi)\right\}\right]\right)d\xi.
\end{equation}
Now by using Eq. \eqref{vim7} we obtain
\begin{equation}\label{vim10}  % {e5}
c_{n+1}(x,t)=c_{n}(x,t)-\int_0^t\left(\frac{\partial{c_n(x,\xi)}}{\partial{\xi}}-\kappa
\ _0\mathcal{D}_\xi^{\alpha,\beta}\left[\frac{\partial^2c_n(x,\xi)}{\partial
x^2}-\frac{\partial}{\partial
x}\left\{u(x)c_n(x,\xi)\right\}\right]\right)d\xi.
\end{equation}
which simplifies to
\begin{equation}\label{vim11}  %{e5}
c_{n+1}(x,t)=c_{n}(x,0)+\kappa\int_0^t \ _0\mathcal{D}_\xi^{\alpha,\beta}\left[\frac{\partial^2c_n(x,\xi)}{\partial
x^2}-\frac{\partial}{\partial
x}\left\{u(x)c_n(x,\xi)\right\}\right]d\xi.
\end{equation}

Starting with an initial approximation
$c_0(x,t)=c(x,0)=f(x)$, we obtain a sequence of successive
approximations, and the exact solution is obtained by taking the
limit of the $n$th approximation, that is,
\begin{equation}\label{vim12}  %{e6}
c(x,t)=\lim_{n\rightarrow \infty}c_{n}(x,t).
\end{equation}

\section{A Case Study}\label{case}

\subsection{Polynomial Uploading} We take $u(x)=-x$
in Eq. \eqref{vim8}, so it becomes
\begin{equation}\label{pol1} %{ee7}
\frac{\partial{c(x,t)}}{\partial{t}} = \kappa
\ _0\mathcal{D}_t^{\alpha,\beta}\left[\frac{\partial^2c(x,t)}{\partial
x^2}+\frac{\partial}{\partial
x}\left\{xc(x,t)\right\}\right],\quad x
> 0, t > 0
\end{equation}
with the initial condition $c(x,0)=x^p$, for $p\geq 0$. We obtain
the following iteration formula by using Eq. \eqref{vim10}
\begin{equation}\label{pol2}  %{e7}
c_{n+1}(x,t)=c_{n}(x,0) + \kappa \int_0^t \ _0\mathcal{D}_\xi^{\alpha,\beta}\left[\frac{\partial^2c_n(x,\xi)}{\partial
x^2}+\frac{\partial}{\partial
x}\left\{xc_n(x,\xi)\right\}\right]d\xi,
\end{equation}
with the zeroth approximation
\begin{equation}\label{pol3}  %{e8}
c_0(x,t)=x^p.
\end{equation}
By taking $n=0$ in Eq. \eqref{pol2} and using Eq. \eqref{pol3}, we
obtain
\begin{equation}\label{pol4}  %{e9}
c_{1}(x,t)=x^p+\kappa\int_0^t \ _0\mathcal{D}_\xi^{\alpha,\beta}\left[\frac{\partial^2}{\partial
x^2}x^p+\frac{\partial}{\partial
x}\left\{x^{p+1}\right\}\right]d\xi,
\end{equation}
which can be written as
\begin{equation}\label{pol5}  %{e10}
c_{1}(x,t)=x^p+\kappa
a_1(x)\int_0^t \ _0\mathcal{D}_\xi^{\alpha,\beta}(1)d\xi,
\end{equation}
where
\begin{equation}\label{pol6}  %{e11}
a_1(x)=\frac{\partial^2}{\partial x^2}x^p+\frac{\partial}{\partial
x}\left\{x^{p+1}\right\}.
\end{equation}
By using Lemma \ref{lma1} and \ref{lma2}, we obtain
%\begin{equation}\label{pol7}   %{e12}
%\mathcal{D}_t^{\alpha,\beta}(t^{\lambda})=\frac{\Gamma(1+\lambda)}{\Gamma(1+\lambda-\alpha)}t^{\lambda-\alpha},
%\end{equation}
%
%Eq. \eqref{pol5} gives
\begin{equation}\label{pol8}  %{e13}
c_{1}(x,t)=x^p+\kappa
a_1(x)\frac{t^{1-\alpha}}{\Gamma(1-\alpha+1)}.
\end{equation}

Importantly, note that for the functions of type $ f(t) = t^\nu$, where $\nu > -1$,
the Hilfer fractional derivative is independent of $\beta$ by Lemma \ref{lma1}.
Moreover, for such functions, the Caputo, Reimann-Liouville and Hilfer
derivatives are all equal.

Furthermore, because this is the first term in the recurrence
relation, $\beta$ does not appear in any of the higher order terms
(below).

By taking $n=1$ in Eq. \eqref{pol2} and using Eq. \eqref{pol8}, we
obtain
\begin{equation}\label{pol9}  %{e14}
c_{2}(x,t)=x^p+\kappa
a_1(x)\frac{t^{1-\alpha}}{\Gamma(1-\alpha+1)}+\kappa^2
a_2(x)\frac{t^{2(1-\alpha)}}{\Gamma(2(1-\alpha)+1)},
\end{equation}
where
\begin{equation}\label{pol10}   % {e15}
a_2(x)=\frac{\partial^2}{\partial
x^2}a_1(x)+\frac{\partial}{\partial x}\left\{xa_1(x)\right\}.
\end{equation}
By taking $n=2$ in Eq. \eqref{pol2} and using Eq. \eqref{pol9}, we
obtain
\begin{equation}\label{pol11}  %{e16}
c_{3}(x,t)=x^p+\kappa
a_1(x)\frac{t^{1-\alpha}}{\Gamma(1-\alpha+1)}+\kappa^2
a_2(x)\frac{t^{2(1-\alpha)}}{\Gamma(2(1-\alpha)+1)}+\kappa^3
a_3(x)\frac{t^{3(1-\alpha)}}{\Gamma(3(1-\alpha)+1)},
\end{equation}
where
\begin{equation}\label{pol12}  % {e17}
a_3(x)=\frac{\partial^2}{\partial
x^2}a_2(x)+\frac{\partial}{\partial x}\left\{xa_2(x)\right\}.
\end{equation}
Proceeding in this way we obtain
\begin{equation}\label{pol13}  %{e18}
c_{n}(x,t)=x^p+\Sigma_{k=1}^{n}\kappa^k
a_k(x)\frac{t^{k(1-\alpha)}}{\Gamma(k(1-\alpha)+1)},
\end{equation}
where %%%%%%%%%%%%%%%%%%%%%%%%%%%%%%%%%%%
\begin{equation}\label{pol14}   %{e19}
a_k(x)=\frac{\partial^2}{\partial
x^2}a_{k-1}(x)+\frac{\partial}{\partial
x}\left\{xa_{k-1}(x)\right\},
\end{equation}
and $a_1(x)$ is given by the Eq. \eqref{pol6}. By setting
$a_0(x)=x^p$, Eq. \eqref{pol13} can be written as
\begin{equation}\label{pol15}  %{e20}
c_{n}(x,t)=\Sigma_{k=0}^{n}
a_k(x)\frac{\kappa^kt^{k(1-\alpha)}}{\Gamma(k(1-\alpha)+1)}.
\end{equation}
By taking the limit $n\rightarrow\infty$ of Eq. \eqref{pol15} we
obtain
\begin{equation}\label{pol16}  % {e21}
c(x,t)=\lim_{n\rightarrow\infty}c_{n}(x,t)=\Sigma_{k=0}^{\infty}
a_k(x)\frac{\kappa^kt^{k(1-\alpha)}}{\Gamma[k(1-\alpha)+1]}.
\end{equation}

We remark, again, that the final solution above does not contain
any dependency on $\beta$. This is a consequence of the Lemma \ref{lma1}.

\subsection{On the Convergence of $c(x,t)$}\label{sec4p2}

\begin{theorem}
  The series solution \eqref{pol16} of the problem \eqref{pol1} with the
initial condition $c(x,0) = x^p$, $p\geq 0$, converges absolutely
  for all $x$ and $t$.
\end{theorem}

\begin{proof}
We denote the $n$th term of Eq. \eqref{pol16} by
$$s_n(x,t)=a_n(x)\frac{\kappa^nt^{n(1-\alpha)}}{\Gamma[n(1-\alpha)+1]}.$$
Applying the ratio test on the series \eqref{pol16}, we obtain
\begin{align}\label{con1}  %{e22}
\nonumber \left|\frac{s_{n+1}(x,t)}{s_n(x,t)}\right|&=\left|\frac{
a_{n+1}(x)}{a_{n}(x)}\kappa t^{(1-\alpha)}\frac{\Gamma[n(1-\alpha)+1]}{\Gamma[(n+1)(1-\alpha)+1]}\right|\\
\nonumber &\text{since } (1-\alpha)>0 \text{ therefore } n(1-\alpha)+1>1\\
&=\left|\frac{ a_{n+1}(x)}{a_{n}(x)}\kappa
t^{(1-\alpha)}\right|\frac{\Gamma[n(1-\alpha)+1]}{\Gamma[(n+1)(1-\alpha)+1]}.
\end{align}
Note that $\left|\frac{ a_{n+1}(x)}{a_{n}(x)}\right|$ is bounded
above by $p+1$, where $p\geq 0$ is the integer power of $x$ in the
initial condition $c(x,0)=x^p$. Indeed, $a_n(x)$, defined in Eq.
\eqref{pol14}, is a polynomial in $x$ whose leading term, that is,
the term with the highest power of $x$ is $(p+1)^nx^p$, and
further note that degree$(a_n(x))=p$ for all $n\geq 0$. Thus we
can approximate $a_n(x)$ by its leading term $(p+1)^nx^p$ (since
all the coefficients are positive) and therefore we obtain
$$\left|\frac{
a_{n+1}(x)}{a_{n}(x)}\right|\approx\left|\frac{(p+1)^{n+1}x^p}{(p+1)^{n}x^p}\right| = p+1.$$
By using Wendel's double inequality, see \cite{feng},
$$x^{1-s}\leq\frac{\Gamma(x+1)}{\Gamma(x+s)}\leq(x+s)^{1-s},$$
for $x>0$ and $0<s<1$, we deduce that
$$\lim_{n\rightarrow\infty}\frac{\Gamma[n(1-\alpha)+1]}{\Gamma[(n+1)(1-\alpha)+1]}=0.$$
Hence, Eq. \eqref{con1} gives
\begin{align}\label{con2}  %{e23}
\lim_{n\rightarrow\infty}\left|\frac{s_{n+1}(x,t)}{s_n(x,t)}\right|=0.
\end{align}
Thus the series solution obtained in Eq. \eqref{pol16} converges
(absolutely) for all $x$ and $t$ and for all real $p$.
\end{proof}

\subsection{To show that $c(x,t)$ obtained in Eq.
\eqref{pol16} satisfies Eq. \eqref{pol1}}

\begin{theorem}
  The series solution \eqref{pol16} satisfies the equation \eqref{pol1}
  with the initial condition $c(x,0) = x^p$, where $p\geq 0$.
\end{theorem}

\begin{proof}
First, on differentiating  Eq. \eqref{pol16} with respect to $t$,
we
obtain%
\begin{align}\label{shw1}   %{e24}
\nonumber \frac{\partial}{\partial t} c(x,t) & =
\frac{\partial}{\partial t}\left[a_0(x) + \Sigma_{k=1}^{\infty} a_k(x) \frac{\kappa^kt^{k(1-\alpha)}}{\Gamma[k(1-\alpha)+1]}\right]\\
\nonumber &=\Sigma_{k=1}^{\infty}
a_k(x)\frac{\kappa^k}{\Gamma[k(1-\alpha)+1]}\frac{\partial}{\partial
t}t^{k(1-\alpha)}\\
\nonumber &=\Sigma_{k=1}^{\infty}
a_k(x)\frac{\kappa^k}{k(1-\alpha)\Gamma[k(1-\alpha)]}k(1-\alpha)t^{k(1-\alpha)-1}\\
 &=\Sigma_{k=1}^{\infty}
a_k(x)\frac{\kappa^kt^{k(1-\alpha)-1}}{\Gamma[k(1-\alpha)]}.
\end{align}
On the other hand substituting Eq. \eqref{pol16} in the right hand
side of Eq. \eqref{pol1} yields
\begin{align}\label{shw2}   %{e25}
\nonumber &\kappa
\ _0\mathcal{D}_t^{\alpha,\beta}\left[\frac{\partial^2 c(x,t)}{\partial
x^2}+\frac{\partial}{\partial x}\left\{x c(x,t)\right\}\right] \\
\nonumber
&=\kappa\ _0\mathcal{D}_t^{\alpha,\beta}\left[\Sigma_{k=0}^{\infty}\left\{\frac{\partial^2}{\partial
x^2}a_k(x)+\frac{\partial}{\partial
x}x^qa_k(x)\right\}\frac{\kappa^kt^{k(1-\alpha)}}{\Gamma[k(1-\alpha)+1]}\right]
\\
\nonumber
&=\kappa\Sigma_{k=0}^{\infty}\left\{\frac{\partial^2}{\partial
x^2}a_k(x)+\frac{\partial}{\partial
x}x^qa_k(x)\right\}\ _0\mathcal{D}_t^{\alpha,\beta}\frac{\kappa^kt^{k(1-\alpha)}}{\Gamma[k(1-\alpha)+1]}
\\
\nonumber
&=\Sigma_{k=0}^{\infty}a_{k+1}(x)\frac{\kappa^{k+1}t^{(k+1)(1-\alpha)-1}}{\Gamma[(k+1)(1-\alpha)]}
\\
&=\Sigma_{k=1}^{\infty}a_{k}(x)\frac{\kappa^{k}t^{k(1-\alpha)-1}}{\Gamma[k(1-\alpha)]}.
\end{align}
The equality of Eqs. \eqref{shw1} and \eqref{shw2} proves that
$u(x,t)$ given by Eq. \eqref{pol16} is the solution of problem
\eqref{pol1}.
\end{proof}

\subsection{Examples}\label{sec4p4}

\begin{figure}[htbp]
  \begin{minipage}[b]{0.5\linewidth}
    \centering
    \includegraphics[width=\linewidth]{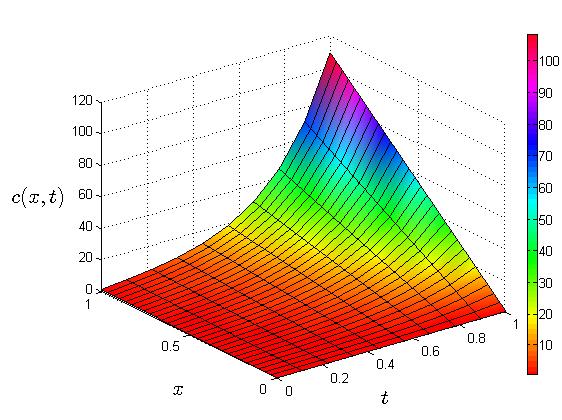}
    \caption{Plot of the solution $c(x,t)$, Eq. \eqref{exp1} (Example 1),
  for $0\leq x\leq 1$ and time $0 <t<1$, where $p=1$, $\alpha=0.5$ and $\kappa=1$.}
    \label{fig1}
  \end{minipage}
  \hspace{1cm}
  \begin{minipage}[b]{0.5\linewidth}
    \centering
    \includegraphics[width=\linewidth]{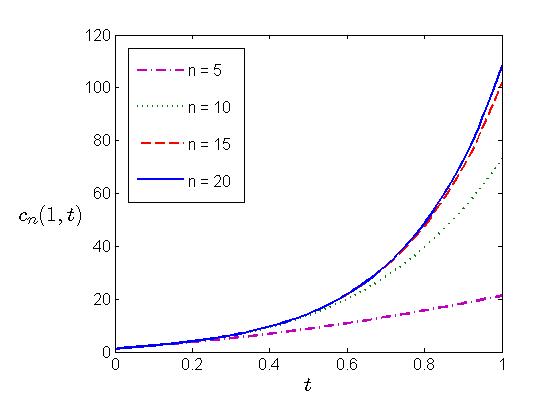}%{figure2.jpg}
    \caption{Plots of the truncated series solution $c_n(x,t)$, Eq. \eqref{exp2},
   (Example 1)  for $0<t<1$ at $x=1$, where $p=1$, $\alpha=0.5$ and $\kappa=1$, for different $n$ as indicated.}
    \label{fig2}
  \end{minipage}
\end{figure}

We examine the solutions for the cases $p=1$ and $2$.

For $p=1$, the initial condition becomes $c(x,0)=x$, in Eq.
\eqref{pol1}.  Then $a_k(x)=2^kx$ for $k\geq 0$ and hence from Eq.
\eqref{pol16} the solution $c(x,t)$ is expressed as follows
\begin{equation}\label{exp1}  %{e26}
c(x,t)=x\Sigma_{k=0}^{\infty}\frac{[2\kappa
t^{1-\alpha}]^k}{\Gamma[k(1-\alpha)+1]}=xE_{1-\alpha}[2\kappa
t^{1-\alpha}],
\end{equation}
where $E_{\alpha}(t)$, defined in Eq. \eqref{pre7}, is the
Mittag-Leffler function in one parameter. The plot of the solution
\eqref{exp1} is shown in the Fig. \ref{fig1} for the values
$\alpha=0.5$ and $\kappa=1$. Note that the solution $c(x,t)$, for
fixed $t$,  increases linearly with respect to variable $x$ and it
increases exponentially with respect to variable $t$, for fixed
$x$. \vskip 0.2cm \noindent In order to see, how rapidly the
sequence of successive approximations provided by VIM converges to
the exact solution, we use the $n$th partial sum as an
approximation,
\begin{equation}\label{exp2}   %{e26a}
c_n(x,t)=x\Sigma_{k=0}^{n}\frac{[2\kappa
t^{1-\alpha}]^k}{\Gamma[k(1-\alpha)+1]},
\end{equation}
In Fig. \ref{fig2}, we plot $c_n(x,t)$ against $t$ at $x=1$ for
$\alpha=0.5$ and $\kappa=1$, for different values of $n$.  One can
see from Fig. \ref{fig2} that the solution converges by $n=20$.
Later in Section \ref{num}, we will provide details about how many
terms have to be summed up in order to obtain a given accuracy.

\begin{figure}[htbp]
  \begin{minipage}[b]{0.5\linewidth}
    \centering
    \includegraphics[width=\linewidth]{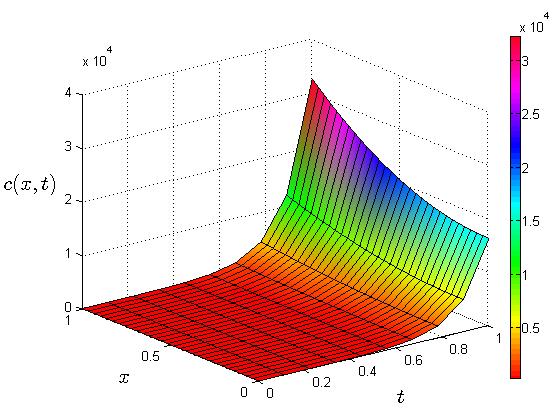}
    \caption{Plot of the solution $c(x,t)$, Eq. \eqref{exp3} (Example 2),
  for $0\leq x\leq 1$ and time $0 <t<1$, where $p=2$, $\alpha=0.5$ and $\kappa=1$.}
  \label{fig3}
 \end{minipage}
  \hspace{1cm}
  \begin{minipage}[b]{0.5\linewidth}
    \centering
    \includegraphics[width=\linewidth]{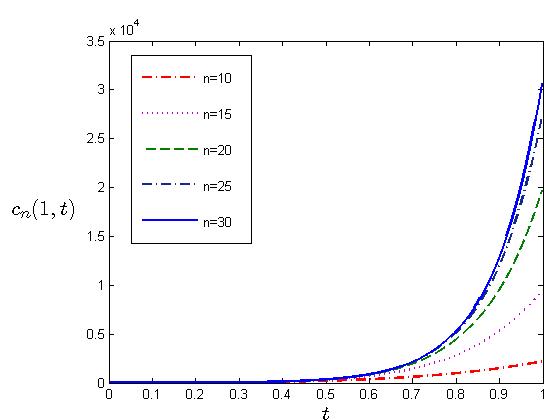}
    \caption{Plots of the truncated series solution $c_n(x,t)$, Eq. \eqref{exp4},
    (Example 2),  for $0<t<1$ at $x=1$, where $p=2$, $\alpha=0.5$ and $\kappa=1$, for different $n$ as indicated.}
    \label{fig4}
  \end{minipage}
\end{figure}

For $p=2$ the initial condition becomes $c(x,0)=x^2$, in Eq.
\eqref{pol1}. Then $a_k(x)=3^kx^2+3^k-1$ for $k\geq 0$ and hence
from Eq. \eqref{pol16} the solution $c(x,t)$ is expressed as
follows
\begin{equation}\label{exp3}   %{e27}
c(x,t)=x^2E_{1-\alpha}[3\kappa t^{1-\alpha}]+E_{1-\alpha}[3\kappa
t^{1-\alpha}]-E_{1-\alpha}[\kappa t^{1-\alpha}].
\end{equation}
The plot of the solution \eqref{exp3} is shown in the Fig.
\ref{fig3} for $\alpha=0.5$ and $\kappa=1$.  This time, the
solution $c(x,t)$, for fixed $t$,  increases quadratically with
respect to variable $x$ and  it increases exponentially with
respect to variable $t$, for fixed $x$.

\vskip 0.2cm \noindent Again, we use the $n$th approximation in
order to see how rapidly the sequence of successive approximations
provided by VIM converges to the exact solution:
\begin{align}\label{exp4}   %{e27a}
\nonumber c_n(x,t)&=x^2\Sigma_{k=0}^{n}\frac{[3\kappa
t^{1-\alpha}]^k}{\Gamma[k(1-\alpha)+1]}+\Sigma_{k=0}^{n}\frac{[3\kappa
t^{1-\alpha}]^k}{\Gamma[k(1-\alpha)+1]}\\
&\quad -\Sigma_{k=0}^{n}\frac{[\kappa
t^{1-\alpha}]^k}{\Gamma[k(1-\alpha)+1]},
\end{align}
and plot it for different values of $n$. Figure \ref{fig4} shows
the plots of $c_n(x,t)$ against $t$ for $x=1$, $\alpha=0.5$,
$\kappa=1$, and for different $n$ as indicated. This time the
approximate solutions converge at $n=25$.

When $p\geq3$, closed form solutions for $c(x,t)$ becomes
increasingly harder to obtain. Nevertheless,  for the purposes of
analyzing the behavior of the solution $c(x,t)$ we require only
the dominant term in the solution. As mentioned in Section
\ref{sec4p2} that we can approximate $a_n(x)$ by its leading term,
that is by $(p+1)^nx^p$. If we replace $a_n(x)$ by $(p+1)^nx^p$ in
Eq. \eqref{pol16}, we obtain the the leading term to be,
\begin{align}\label{exp5}    %{e28}
\nonumber c(x,t)&\approx\Sigma_{k=0}^{\infty}(p+1)^nx^p
\frac{\kappa^kt^{k(1-\alpha)}}{\Gamma[k(1-\alpha)+1]}\\
\nonumber &\approx x^p\Sigma_{k=0}^{\infty} \frac{[(p+1)\kappa
t^{1-\alpha}]^k}{\Gamma[k(1-\alpha)+1]}\\
&\approx x^pE_{1-\alpha}[(p+1)\kappa t^{1-\alpha}].
\end{align}
The solution $c(x,t)$ is thus proportional to $x^p$ at fixed t;
and it increases approximately exponentially with respect to
variable $t$ at fixed $x$.

In order to further investigate the trends in the fractional
solution, we compare the fractional solution to the non-fractional
solution at $x=1$ and $t=1$.

 At $x=1$ and $t=1$, we get $c(1,1)\approx
E_{1-\alpha}[(p+1)\kappa]$, and by using Eq. \eqref{pre7} we
obtain
$$c(1,1)\approx
\sum_{k=0}^{\infty}\frac{[(p+1)\kappa]^k}{\Gamma[k(1-\alpha)+1]}.$$
In the asymptotic limit $\alpha\rightarrow 1$, we obtain $c(1,1)\approx
\sum_{k=0}^{\infty}[(p+1)\kappa]^k$, which is a geometric series
and it converges to $\frac{1}{1-(p+1)\kappa}$, when
$(p+1)\kappa<1$ or $\kappa<\frac{1}{p+1}$.

\vskip 0.2cm \noindent For $\alpha=0$, we obtain
$$c(1,1)\approx \sum_{k=0}^{\infty}\frac{[(p+1)\kappa]^k}{k!}$$
which converges for all $\kappa$ and $p\geq 0$.

Furthermore, in Eq. \eqref{exp5} in the asymptotic limit
$\alpha\rightarrow 1$,  $t^{1-\alpha}= 1$  for all $t$, and thus
the power series \eqref{exp5} converges for all $t$ so long as
$(p+1)\kappa < 1$.

\section{Numerical Analysis and behavior of the
solution}\label{num}

In this section we discuss the general behavior of the solution of
the problem \eqref{pol1} with respect to different parameters and
also we do some numerical analysis of the solution.

\subsection{Reciprocal Gamma Function} %\vskip 0.2cm \par

 First we analyze the reciprocal gamma function
which appears in the series solution \eqref{pol16}, that is,
\begin{equation*}
\frac{1}{\Gamma[n(1-\alpha)+1]}.
\end{equation*}
The limit of this function, for a fixed $\alpha(\neq 1)$, as $n
\to \infty$ is zero, that is,
\begin{equation*}
\lim_{n \to \infty}\frac{1}{\Gamma[n(1-\alpha)+1]}=0, \quad \text{
for fixed } \alpha\in(0,1).
\end{equation*}
We want an expression for the number of terms $n$ needed in
order to satisfy a given accuracy given by,
\begin{equation*}
\frac{1}{\Gamma[n(1-\alpha)+1]}\leq 10^{-\tau},
\end{equation*}
 for a given $\alpha$ and
tolerance level $\tau$. For this purpose, we use the following
formula, see \cite{alzr},
\begin{align}
\nonumber &\sqrt{2\pi x}(x/e)^x(x\sinh 1/x)^{x/2}(1+a/x^5) <
\Gamma(x+1)\\
\nonumber &< \sqrt{2\pi x}(x/e)^x(x\sinh 1/x)^{x/2}(1+b/x^5),
\end{align}
for all $x>0$ and with the optimal constants $a=0$ and
$b=\frac{1}{1620}$, to obtain the required $n$. Table \ref{Ta:1}
summarizes this information.

Figure \ref{fig6} shows the plots of $n$ against $\tau$ for
different values of $\alpha$.  For any given $\alpha$ the number
of terms $n$ appears to scale almost  linearly with $\tau$. Best
linear fits were therefore obtained; for example for $\alpha=0.9$,
we obtain the best fit $n=40.79+8.964\tau$. The graph of this
linear equation is the solid line shown in the Fig. \ref{fig6}
along with the $95\%$ prediction intervals, whereas the simulation
data are plotted as the symbols. The value of $R^2$, the coefficient
of determination, is $99.8\%$ represents the percent of observed
variability explained by the linear model, whereas the value of
$R_{adj}^2$ is $99.7\%$ which is a more realistic quantity as it
accounts for the number of terms in the model.

\begin{table}[t]
   % \caption{Table caption}\label{tab:booktabs}
    \centering
    \begin{tabular}{*{11}{c}}\toprule
      \multirow{2}{*}[-0.5ex]{$\tau$} & \multicolumn{10}{c}{$\alpha$} \\ \cmidrule{2-11}
       &$0$ & $0.1$ & $0.2$ & $0.3$ & $0.4$ & $0.5$ & $0.6$ &
$0.7$ & $0.8$ & $0.9$\\ \midrule $4$ &$8$ & $9$ & $10$ & $11$
&$13$ & $15$ & $19$ & $25$ & $37$ &
$74$ \\
$6$ &$10$ & $11$ & $12$ & $14$ & $16$ & $19$ & $24$ &
$32$ & $48$ & $95$ \\
$8$ &$12$ & $13$ & $15$ & $17$ & $19$ & $23$ & $29$ & $38$ &
$57$ & $114$ \\
      $10$ &$14$ & $15$ & $17$ & $19$ & $22$ & $27$ & $33$ &
$44$ & $66$ & $132$ \\
    $12$ &$15$ & $17$ & $19$ & $22$ & $25$ & $30$ & $38$ &
$50$ & $75$ & $150$ \\
  $14$ &$17$ & $19$ & $21$ & $24$ & $28$ & $34$ & $42$ &
$56$ & $83$ & $166$ \\
      $16$ & $19$ & $21$ & $23$ & $26$ & $31$ & $37$ & $46$ &
$61$ & $91$ & $182$ \\  \bottomrule
    \end{tabular}
    \vskip 0.4cm
    \caption{The number of terms, $n$, needed to achieve a given tolerance level $10^{-\tau}$ for different $\alpha$.}\label{Ta:1}
  \end{table}
\begin{figure}[htbp]
  \begin{minipage}[b]{0.5\linewidth}
    \centering
    \includegraphics[width=\linewidth]{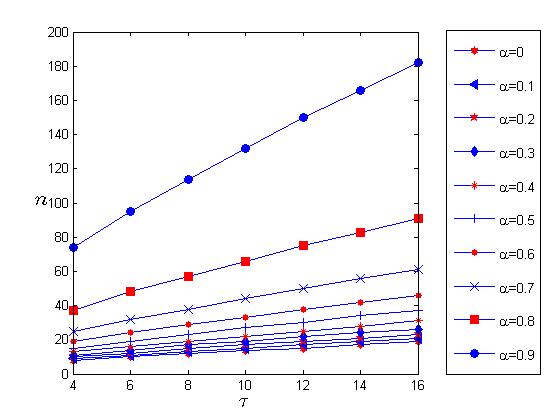}\\
    \caption{Plots of $n$ against $\tau$, for specific values of $\alpha$ }\label{fig5}
  \end{minipage}
%  \hspace{0.2cm}
  \begin{minipage}[b]{0.5\linewidth}
    \centering
    \includegraphics[width=\linewidth]{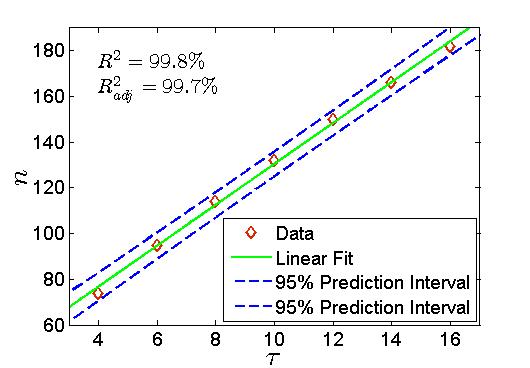}\\
     \caption{A linear relationship is observed between $\tau$ and $n$. Here, $\alpha=0.9$.}\label{fig6}
  \end{minipage}
\end{figure}

%\begin{figure}[htbp]
%  \begin{minipage}[b]{0.4\linewidth}
%    \centering
%    \includegraphics[width=\linewidth]{figure6.jpg}\\
%     \caption{Plots of $n$ against $\tau$, for specific values of
%     $\alpha$.}\label{fig6} %5a}
%  \end{minipage}
%
%  \begin{minipage}[b]{0.5\linewidth}
%    \centering
% \includegraphics[width=\linewidth]{Fig7.jpg}\\
%  \caption{A linear relationship is observed between $\tau$ and $n$. Here,
%   $\alpha=0.9$.}\label{fig7}  %5b}
%  \end{minipage}
%\end{figure}

%\begin{figure}[t]
% \includegraphics[width=12cm]{figure7.jpg}\\
%  \caption{A linear relationship is observed between $\tau$ and $n$. Here,
%   $\alpha=0.9$.}\label{fig7}  %5b}
%\end{figure}

Figure \ref{fig8} shows the plots of $n$ against $\alpha$ for some
tolerance levels, taken from  the data in Table \ref{Ta:1}. The
relationship between $n$ and $\alpha$ is clearly non-linear.

For $\alpha\in(0,0.7)$  the number of terms
needed for a given tolerance level $10^{-\tau}$ remains
approximately constant; in fact even for different tolerance
levels $n$ appears to be approximately insensitive to $\alpha$ and
to $\tau$ -- as a rule of thumb we see that $n\approx 20$ for
$\alpha<0.7$.

But for $\alpha> 0.7$ we see a rapid increase in the number of
terms $n$ needed to achieve a given accuracy.  It may be possible
to find best-fit curves to the data plotted in Fig. \ref{fig8}.
For this purpose, we assume cubic polynomial fits of the form
$n(\alpha)=A+B\alpha+C\alpha^2+D\alpha^3$, where $A,B,C$ and $D$
are constants to be determined from the data using the least
square method,  for a tolerance level of $10^{-16}$, we obtain the
following cubic polynomial
$n(\alpha)=13.69+181.3\alpha-665.6\alpha^2+730.4\alpha^3$. The
graph of this cubic polynomial is the solid line shown in the Fig.
\ref{fig9} along with the $95\%$ prediction interval, whereas the 
simulation data are plotted symbols. The value of $R^2$, is $97.5\%$, 
and $R_{adj}^2$ is $96.2\%$.
\begin{figure}[htbp]
  \begin{minipage}[b]{0.5\linewidth}
    \centering
    \includegraphics[width=\linewidth]{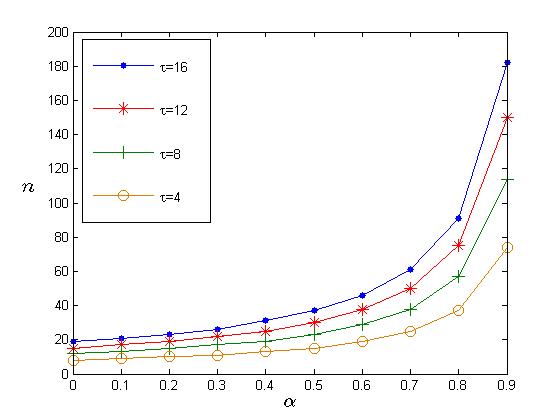}\\
    \caption{Plots of $n$ against $\alpha$, for specific tolerance levels
  $10^{-\tau}$.}\label{fig8}  %6}
  \end{minipage}
 % \hspace{1cm}
  \begin{minipage}[b]{0.5\linewidth}
    \centering
    \includegraphics[width=\linewidth]{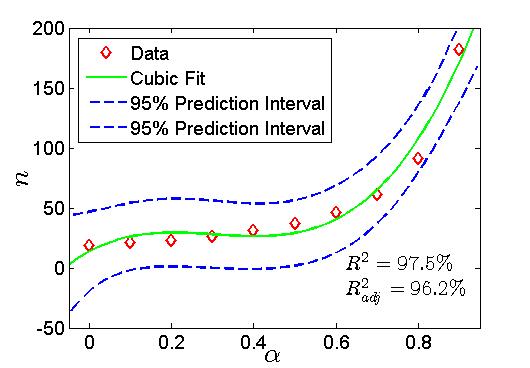}\\
    \caption{ A cubic relationship is observed between $\alpha$ and $n$. Here,
  the tolerance level is $10^{-16}$.}\label{fig9}  %7}
  \end{minipage}
\end{figure}

  \begin{table}[!ht]
   % \caption{Table caption}\label{tab:booktabs}
    \centering
    \begin{tabular}{*{12}{c}}\toprule

      \multirow{2}{*}[-0.5ex] {$p$} &  \multirow{2}{*}[-0.5ex] { $\kappa$} & \multicolumn{10}{c}{$\alpha$} \\

\cmidrule{3-12}
   &    &$0$ & $0.1$ & $0.2$ & $0.3$ & $0.4$ & $0.5$ & $0.6$ &
$0.7$ & $0.8$ & $0.9$\\

\midrule \addlinespace \\

1&$0.1$ &$1.2214$ & $1.2339$ & $1.2456$ &  $1.2563$ &
 $1.2655$ & $1.2726$ & $1.2769$ & $1.2777$ & $1.2740$ & $1.2650$ \\
 & $0.2$ &$1.4918$ & $1.5288$ & $1.5666$ &  $1.6046$ &
 $1.6417$ & $1.6762$ & $1.7056$ & $1.7257$ & $1.7311$ & $1.7140$ \\
 &$0.3$ &$1.8221$ & $1.9018$ & $1.9891$ &  $2.0845$ &
 $2.1881$ & $2.2989$ & $2.4132$ & $2.5225$ & $2.6064$ & $2.6228$ \\
 & $0.4$ &$2.2255$ & $2.3745$ & $2.5486$ &  $2.7548$ &
 $3.0024$ & $3.3039$ & $3.6758$ & $4.1366$ & $4.6886$ & $5.2181$ \\
& $0.5$ &$2.7183$ & $2.9749$ & $3.2946$ &  $3.7041$ &
 $4.2486$ & $5.0090$ & $6.1471$ & $8.0407$ & $11.8230$ & $23.1605$ \\
& $0.6$ &$3.3201$ & $3.7394$ & $4.2952$ &  $5.0672$ &
 $6.2094$ & $8.0628$ & $11.5427$ & $20.0518$ & $58.6544$ & $916.0995$ \\

\midrule \addlinespace \\
2&$0.1$ &$1.5945$ & $1.6353$ & $1.6758$ & $1.7151$ & $1.7518$ &
$1.7839$ & $1.8084$ &
$1.8219$ & $1.8194$ & $1.7957$ \\
&$0.2$ & $2.4228$ & $2.5697$ & $2.7325$ & $2.9127$ & $3.1107$ &
$3.3251$ & $3.5495$ &
$3.7672$ & $3.9389$ & $3.9806$ \\
&$0.3$ &$3.5693$ & $3.9406$ & $4.3938$ & $4.9582$ & $5.6782$ &
$6.6248$ & $7.9175$ &
$9.7701$ & $12.5841$ & $16.9670$ \\
&$0.33$ &$3.9915$ & $4.4651$ & $5.0595$ & $5.8263$ & $6.8503$ &
$8.2856$ & $10.4192$ &
$13.9340$ & $20.7619$ & $39.3061$ \\
%&$0.34$ &$4.1414$ & $4.6129$ & $5.2735$ & $6.1559$ & $7.3233$ &
%$8.9409$ & $11.4483$ &
%$15.8485$ & $25.1379$ & $58.4301$ \\
 &     $0.4$ &$5.1418$ & $5.9195$ & $7.0055$ & $8.5392$ & $10.7960$ & $14.4464$ & $21.3697$ &
$38.3904$ & $115.8880$ & $9036.86$ \\

\midrule \addlinespace \\
3&$0.1$ &$2.3031$ & $2.4136$ & $2.5300$ & $2.6500$ & $2.7704$ &
$2.8871$ & $2.9914$ &
$3.0700$ & $3.1023$ & $3.0610$ \\
&$0.15$ &$3.2389$ & $3.4888$ & $3.7704$ & $4.0874$ & $4.4422$ &
$4.8342$ & $5.2536$ &
$5.6714$ & $6.0144$ & $6.1219$ \\
&$0.2$ &$4.4267$ & $4.9114$ & $5.4945$ & $6.2053$ & $7.0842$ &
$8.1867$ & $9.5866$ &
$11.3700$ & $13.5612$ & $15.7306$ \\
&$0.25$ &$5.9270$ & $6.7868$ & $7.8886$ & $9.3416$ & $11.3284$ &
$14.1788$ & $18.5512$ &
$25.9800$ & $41.0229$ & $86.4111$ \\
 &     $0.3$ &$7.8141$ & $9.2522$ & $11.2138$ & $14.0152$ & $18.2732$ & $25.3548$ & $38.9312$ &
$72.6400$ & $227.4317$ & $18069.2957$ \\

\midrule \addlinespace \\
4&$0.10$ &$3.6044$ & $3.9008$ & $4.2330$ & $4.6032$ & $5.0110$ &
$5.4495$ & $5.8996$ &
$6.3177$ & $6.6170$ & $6.6440$ \\
&$0.14$ &$5.3248$ & $5.9671$ & $6.7391$ & $7.6743$ & $8.8142$ &
$10.2060$ & $11.8918$ &
$13.8682$ & $15.9573$ & $17.4156$ \\
&$0.18$ &$7.5956$ & $8.8247$ & $10.4021$ & $12.4722$ & $15.2609$ &
$19.1376$ & $24.7363$ &
$33.2074$ & $46.7204$ & $68.5890$ \\
&$0.20$ &$8.9816$ & $10.6300$ & $12.8138$ & $15.7965$ & $20.0255$
& $26.3213$ & $36.3427$ &
$53.9703$ & $90.7746$ & $203.926$ \\
 &     $0.22$ &$10.5624$ & $12.7377$ & $15.7134$ & $19.9434$ & $26.2635$ & $36.7332$ & $54.2600$ &
$91.6757$ & $202.0610$ & $1272.2308$ \\

\bottomrule %
    \end{tabular}
    \vskip 0.5cm
        \caption{The solution $c(1,1)$ for different parameter values $p$,  $\alpha$ and $\kappa$.}\label{Ta:4}
  \end{table}

\subsection{Behavior of the Solution}

The solution of the problem \eqref{pol1} (with $u(x)=-x$, and initial condition $c(x,0)=x^p$, $p>0$) is given by
$$c(x,t)=\Sigma_{n=0}^{\infty}
a_n(x)\frac{(\kappa t^{1-\alpha})^n}{\Gamma[n(1-\alpha)+1]},
$$
where
$$a_n(x)=\frac{\partial^2}{\partial
x^2}a_{n-1}(x)+\frac{\partial}{\partial
x}\left\{xa_{n-1}(x)\right\},
$$ and $a_0(x)=x^p$.

We now examine the sensitivity of the solution $c(x,t)$  to parameters $p$, $\kappa$ the diffusivity coefficient,  and to $\alpha$ the order of the Hilfer derivative.
We choose the following values: $p\in \{1,2,3,4\}$, $\alpha\in\{0,
0.1, ... , 0.9\}$ and we select $\kappa$ according to
$\kappa<\frac{1}{p+1}$. We compare the values of $c(x,t)$ at the
point $(1,1)$ for different combinations of the above parameters.
We choose three values of $\kappa$ that are less than
$\frac{1}{p+1}$ and one value that is equal (approximately) to
$\frac{1}{p+1}$ and one value that is  greater than
$\frac{1}{p+1}$. The data is collected in Table \ref{Ta:4} and it 
reveals that different combinations of parameter's
values affect the values of $c(x,t)$ differently. In general, we
note that an increase in the values of parameters $p$, $\kappa$
and $\alpha$ results into an increase in the values of $c(x,t)$.
The effect of $\kappa$ can be divided into two parts, first when
$\kappa<\frac{1}{p+1}$ and second when $\kappa\geq \frac{1}{p+1}$.

For the first case, $\kappa<\frac{1}{p+1}$, the increase in the
values of $c(x,t)$ is not significant as compare to the second case,
$\kappa\ge\frac{1}{p+1}$, where $c(x,t)$
increases very rapidly with $\alpha$ and $\kappa$.
%
%
%\begin{figure}[t]
% \includegraphics[width=12cm]{figure10.jpg}
%  \caption{Plots of $c(x,t)$ at $(1,1)$ against $\alpha$
%    and $\kappa$, for $p=1,2,3,4$.}\label{fig10}  %9}
%\end{figure}
%%
%
\begin{figure}[htbp]
  \begin{minipage}[b]{0.5\linewidth}
    \centering
    \includegraphics[width=\linewidth]{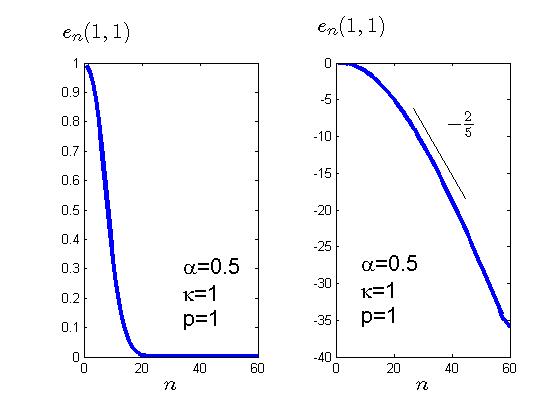}
    \caption{\small{Plot of the relative error $e_n(1,1)$ against
    the number of terms $n$, from example 1. Left: Linear-Linear,
      Right: Linear-Natural Log (y-axis is scaled as $e^{m}$,
      where $m\in\{-40,-35,\dots,0\}$). A line of slope $-2/5$ is shown for comparison.}}
    \label{fig11} %0}
  \end{minipage}
  \hspace{0.2cm}
  \begin{minipage}[b]{0.5\linewidth}
    \centering
    \includegraphics[width=\linewidth]{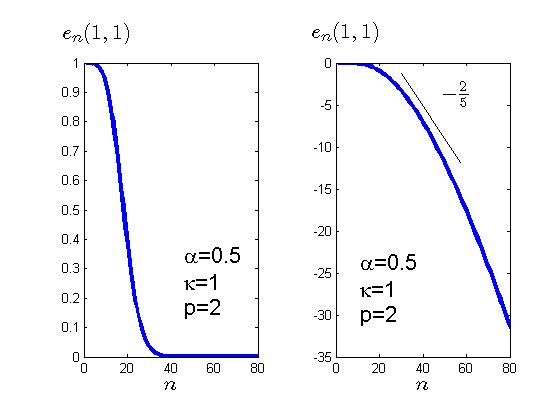}
    \caption{\small{Plot of the relative error $e_n(1,1)$ against the number of terms $n$, from example 2. Left: Linear-Linear,
       Right: Linear-Natural Log (y-axis is scaled as $e^{m}$,
      where $m\in\{-35,-30,\dots,0\}$). A line of slope $-2/5$ is shown for comparison.}}
    \label{fig12}  %1}
  \end{minipage}
\end{figure}

\subsection{Relative Error}\label{RelErr}

Finally, we examine the truncation error when using the partial
sum of first $n$ terms of the series solution given in Eq.
\eqref{pol16}. We define the relative error as,
$$e_n(x,t)=\frac{|c(x,t)-c_n(x,t)|}{|c(x,t)|},$$
where $c_n(x,t)$ is the partial sum of the first $n$ terms.
Figures \ref{fig11} and \ref{fig12} show the plots of $e_n(x,t)$
versus $n$ at the point $(1,1)$ for, respectively, $p=1$ (Example
1), and $p=2$ (Example 2). The errors fall off exponentially fast
with $n$, $e_n\approx \exp(-2/5)$.

It is not always possible to express the series solution in
compact form like in the cases of $p = 1$ and $p = 2$. For  values of $p$
larger than $2$, we approximate the 'exact' solution by taking  a
very large value of $n$,  and then we compare this with the
smaller values of $n$. For example, when the initial condition is
$f(x)=x^3$, $(p=3)$, we take $n=200$,  and $c_{200}(x,t)$ is the
approximation to the exact solution. Figure \ref{fig13} shows the
plot of
$$e_n(x,t)=\frac{|c_{200}(x,t)-c_n(x,t)|}{|c_{200(x,t)}|},$$
for $1\leq n \leq 100$. Again the error falls off exponentially
fast like $e_n\approx \exp(-2/5)$. This indicates that the numerical
convergence is independent of the power exponent $p$.

\begin{figure}[t]
 \includegraphics[width=12cm]{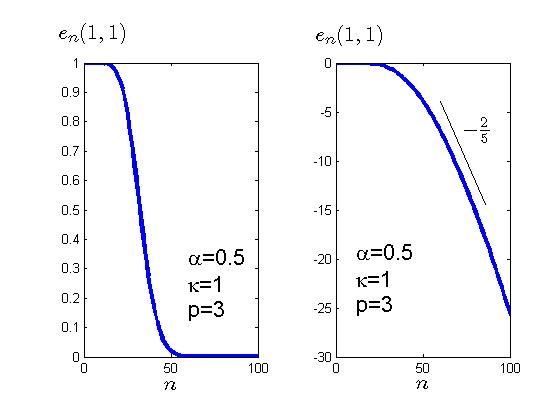}\\
    \caption{\small{Plot of the relative error $e_n(1,1)$ against the number of terms
       $n$, for $p=3$. Left: Linear-Linear,    Right: Linear-Natural Log (y-axis is scaled as $e^{m}$,
      where $m\in\{-30,-25,\dots,0\}$). A line of slope $-2/5$ is shown for comparison.}}
\label{fig13}  %2}
\end{figure}

These results show that the VIM  solutions for this problem
convergence rapidly, and actually improves  as $n$ increases.

\section{Fractional versus Conventional Solutions}\label{FrVsCon}

 In this section, we compare the solutions of the fractional
 differential Eq. \eqref{pol1} with the corresponding
 conventional differential equation. \vskip 0.2cm

 Conventional version of Eq. \eqref{pol1} can be obtained by
taking $\alpha=0$, see \cite{kilbas}, that is,
\begin{equation}\label{comp1}
\frac{\partial{c(x,t)}}{\partial{t}}=\kappa\left[\frac{\partial^2c(x,t)}{\partial
x^2}+\frac{\partial}{\partial
x}\left\{xc(x,t)\right\}\right],\quad x> 0, t > 0.
\end{equation}

The solution of Eq. \eqref{comp1} with initial condition
$c(x,0)=x$ is
\begin{equation}\label{comp2}
c(x,t)=x e^{2\kappa t},
\end{equation}
and the general solution for $\alpha >0$  is given by, Section
\ref{sec4p4},
\begin{equation}\label{comp3}
c(x,t)=x E_{1-\alpha}[2 \kappa t^{1-\alpha}].
\end{equation}
  Figure \ref{fig14} shows the plots of the solution \eqref{comp3}
for different values of $\alpha$, that also includes the case
$\alpha=0$  at $x=1$, for $0<t<1$ and $\kappa=0.4$.

 Thus, the relative magnitude of the solutions compared to the
conventional case can be estimated from,
\begin{equation}\label{comp4}
{c^\alpha(x,t)\over c^0(x,t)}  \approx {E_{1-\alpha} [2 \kappa
t^{1-\alpha}] \over e^{2 \kappa t}}.
\end{equation}
\begin{figure}[htbp]
  \begin{minipage}[b]{0.5\linewidth}
    \centering
    \includegraphics[width=\linewidth]{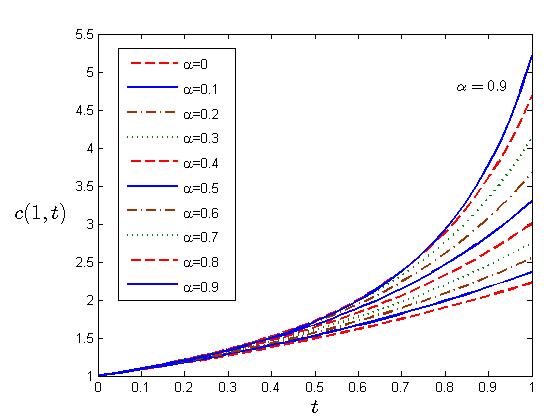}
    \caption{Plot of the solution $c(x,t)$, Eq. \eqref{comp3},
  when $x=1$,  $0 <t<1$ and $\kappa=0.4$.}\label{fig14} %3}
  \end{minipage}
  \hspace{0.5cm}
  \begin{minipage}[b]{0.5\linewidth}
    \centering
    \includegraphics[width=\linewidth]{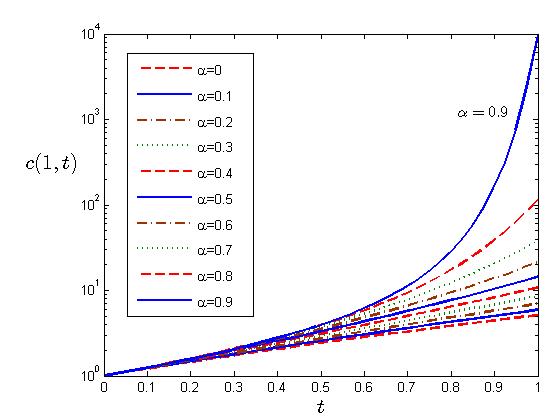}
     \caption{Plot of the solution $c(x,t)$, Eq. \eqref{comp6},
  when $x=1$,  $0 <t<1$ and $\kappa=0.4$.}\label{fig15} %4}
  \end{minipage}
\end{figure}

Asymptotic behavior of the above expression can be analyzed through the
long time behavior of Mittag-Leffler function. By using Theorem 1.3 of
\cite{igor}, we can write
\begin{equation}\label{comp4a}
{c^\alpha(x,t)\over c^0(x,t)}  \approx \exp\{((2\kappa)^{1/(1-\alpha)}-2\kappa)t\}.
\end{equation}
The long time behavior depends on the value of $\kappa$.
If $\kappa < 1/2$ then ${c^\alpha/ c^0}  \to 0$;  if $\kappa = 1/2$ then ${c^\alpha/ c^0}  \to constant$ and the fractional solution scales with the conventional solution; if $\kappa > 1/2$ then ${c^\alpha/ c^0}  \to \infty$.

The solution of Eq. \eqref{comp1} with initial condition
$c(x,0)=x^2$ is
\begin{equation}\label{comp5}
c(x,t)=x^2e^{3\kappa t}+e^{3\kappa t}-e^{\kappa t},
\end{equation}
where as, the fractional solution, Example 2 in Section
\ref{sec4p4}, is
\begin{equation}\label{comp6}
c(x,t)=x^2E_{1-\alpha}[3\kappa t^{1-\alpha}]+E_{1-\alpha}[3\kappa
t^{1-\alpha}]-E_{1-\alpha}[\kappa t^{1-\alpha}].
\end{equation}
Figure \ref{fig15} shows the plots of the solution \eqref{comp6}
for different values of $\alpha$, that also includes the case
$\alpha=0$ which corresponds to the classical case,  at $x=1$, for
$0<t<1$ and $\kappa=0.4$.

The relative magnitude of the solutions compared to the
conventional case can be estimated from,
\begin{equation}\label{comp7}
{c^\alpha(x,t)\over c^0(x,t)}  \approx {x^2E_{1-\alpha}[3\kappa
t^{1-\alpha}]+E_{1-\alpha}[3\kappa
t^{1-\alpha}]-E_{1-\alpha}[\kappa t^{1-\alpha}] \over
x^2e^{3\kappa t}+e^{3\kappa t}-e^{\kappa t}}.
\end{equation}

To leading order, Eq. \eqref{comp7} is,
\begin{equation}\label{comp8}
{c^\alpha(x,t)\over c^0(x,t)}  \approx {E_{1-\alpha}[3\kappa
t^{1-\alpha}] \over e^{3\kappa t}}.
\end{equation}
and is easily shown that for general $p>1$, i.e. for initial conditions
$c(x,0)=x^p$, the corresponding relative magnitude is given by,
\begin{equation}\label{comp9}
{c^\alpha(x,t)\over c^0(x,t)}  \approx {E_{1-\alpha}[(p+1)\kappa
t^{1-\alpha}] \over e^{(p+1)\kappa t}};
\end{equation}
and the large time behaviour is,
\begin{equation}\label{comp9}
{c^\alpha(x,t)\over c^0(x,t)}  \to
\exp\{(((p+1)\kappa)^{1/(1-\alpha)}-(p+1)\kappa)t\} \ \ \ \ {\rm as}\ t\to\infty
\end{equation}

If $\kappa < 1/(1+p)$ then ${c^\alpha/ c^0}  \to 0$;  if $\kappa = 1/(1+p)$ then ${c^\alpha/ c^0}  \to constant$ and the fractional solution scales with the conventional solution; if $\kappa > 1/(1+p)$ then ${c^\alpha/ c^0}  \to \infty$.

\section{Conclusion}\label{conc}

Some complex physical phenomenon such as crowded systems, and transport through porous
media are not fully understood at the present time, and in order
to shed  new light into such phenomena a  new modeling
strategy has emerged in recent years which involves casting the
system of interest in terms of fractional calculus.

In this study the Hilfer fractional advection-diffusion equation of order $0<\alpha<1$ and type $0\leq\beta\leq1$,
\begin{equation*}
\frac{\partial{c(x,t)}}{\partial{t}}=\kappa
\ _0\mathcal{D}_t^{\alpha,\beta}\left[\frac{\partial^2c(x,t)}{\partial
x^2}-\frac{\partial}{\partial x}\left\{u(x)c(x,t)\right\}\right],\quad x > 0, t > 0,
\end{equation*}
with advection term $u(x) = -x$, and power law initial conditions of
the type $c(x,t=0) = x^p$ for $p>0$, was investigated numerically
using the Variational Iteration Method (VIM) method, with a view of
comparing its solution to the conventional non-fractional
advection-diffusion system, and also to analyze the system
numerically in order to investigate the efficiency of solving such
systems numerically.

For this class of initial conditions $c(x,t=0) = x^p$ for $p>0$, the
above problem yields the same solutions as Caputo and Riemann-Liouville
advection-diffusion equations. However, we remark that this would be different
in case if there is a source term in the problem or if the velocity
also involves time variable.

Power series solutions were obtained, Eq. \eqref{pol16}. These
yield closed form solutions for specific $p>0$ in terms of the
Mittag-Leffler functions, although it becomes increasingly
difficult to actually calculate it for $p>2$. Nevertheless, the
leading order term can readily  be obtained even for $p>2$, which
allows us to investigate the asymptotic behavior of the solutions
and to carry out some numerical analysis of the method used. The
power series is (absolutely) convergent (for all $x$, and all $t$,
and for $p\geq 0$).

The behavior of the solution was examined by comparing the value
of the solution at a fixed point, namely $|c(1,1)|$.
Asymptotically, the relative magnitude of the solutions is,
$|{c^\alpha(x,t)\over c^0(x,t)}|  \approx
{E_{1-\alpha}({(p+1)\kappa t}) \over e^{(p+1)\kappa t}}$ which
shows that the fractional solution increases approximately
exponential faster than the conventional solution, as seen in
Figs. \ref{fig14} and \ref{fig15}.

 For fixed $p$, the increase in $|c(1,1)|$,  when
$\alpha\in (0,0.7)$ and $\kappa<\frac{1}{p+1}$,  is small; but
when $\alpha> 0.7$ and $\kappa\geq\frac{1}{p+1}$,  $|c(1,1)|$
increases very rapidly. For fixed $\alpha$ and $\kappa$, the
solution increases polynomially with $x$ and exponentially with
$t$. 

For the long time behaviour, as $t\to\infty$, we find that for all $x$ and for all $0<\alpha<1$, if $\kappa < 1/(1+p)$ then ${c^\alpha/ c^0}  \to 0$; and  if $\kappa = 1/(1+p)$ then ${c^\alpha/ c^0}  \to constant$ and the fractional solution scales with the conventional solution; and if $\kappa > 1/(1+p)$ then ${c^\alpha/ c^0}  \to \infty$.

Variational iteration method (VIM) has proved to be an efficient
method for obtaining the series solution of the Hilfer fractional
advection-diffusion equation with the given power law initial
data. Truncation errors $\epsilon(n)$, arising when using the
partial sum as approximate solutions, decay exponentially fast
with the number of terms $n$ used, and then rate of convergence is
independent of $p$ for the cases considered, $\epsilon(n)\sim
\exp(- 2/5)$, for $p=1,2,3$.

The number of terms $n$ required for a given level of accuracy for
$\alpha< 0.7$ are relatively insensitive to both the $\alpha$ and
to the accuracy level required; but for  $\alpha>0.7$ the number
of terms increases rapidly with $\alpha$ and with the accuracy level
required. This threshold $\alpha\approx 0.7$ is consistent with
the analysis of the solutions $|c(1,1)|$ above.

Although these are early days in the development of fractional calculus and numerical
solutions to fractional equations that describe physical systems, it is clear that
numerical methods like VIM will be important tools in extracting solutions of such fractional PDE's in the future.

Future work will address the case when we have non-zero initial
conditions which should yield different solutions for each
$\beta$.

%%%%%%%%%%%%%%%%%%%%%%%%%%%%%%%%%%%%%%%%%%%%%%%%%%%%%%%%%%%%%%%%%%%%%%%%%%%%%%%%%%%%

\section*{Acknowledgements} The authors would like to thank NSTIP
for funding through project number 11-OIL1663-04. We also thank to
the ITC department at KFUPM for providing software assistance. 

\section*{Appendix A}
In a single-phase single-component system, let us denote the
concentration of the scalar by $c(x,t)$ at the position $x$ and at
the time instant $t$. First, if $u(x,t)$ is the velocity
field in which the scalar is transported, then the advection
equation (without reaction or diffusion) is \cite{hund},
\begin{equation}\label{eq1}
\frac{\partial{c(x,t)}}{\partial{t}}+\frac{\partial}{\partial
x}\left\{u(x,t)c(x,t)\right\} = 0.
\end{equation}

Second, if the concentration $c(x,t)$ is transported by
diffusion and these changes are caused by gradients in $c(x,t)$
and the fluxes across the boundaries of the region, then the diffusion
equation (without advection or reaction) is \cite{hund},
\begin{equation}\label{eq2}
\frac{\partial c}{\partial{t}}=-\frac{\partial}{\partial
x}\left(J_c\right),
\end{equation}
where $J_c=-d(x,t)\frac{\partial c}{\partial{x}}$ is the flux of
$c$ at the point $(x,t)$, and $d(x,t)$ is the coefficient of
diffusivity. This form in which the flux $J_c$ is proportional to
the scalar gradient is called Fickian diffusion.

Finally, there may be local changes in $c(x,t)$ due
to sources, sinks, and chemical reactions, which is described by
an additional source (reaction) term $f(x,t,c(x,t))$. The reaction
equation (without advection or diffusion) is \cite{hund},
\begin{equation}\label{eq3}
\frac{\partial{c(x,t)}}{\partial{t}} = f(x,t,c(x,t)).
\end{equation}

The general advection-diffusion-reaction equation is obtained by
combining the above three effects and the overall change in the
concentration $c(x,t)$ is described by the following equation \cite{hund}
\begin{equation}\label{eq4}
\frac{\partial{c(x,t)}}{\partial{t}}+\frac{\partial}{\partial
x}\left(u(x,t)c(x,t)\right)=-\frac{\partial}{\partial
x}\left(J_c\right)+f(x,t,c(x,t)).
\end{equation}

Applications of such equations arise in many fields, such as,
atmospheric chemistry, \cite{crut}, air pollution models,
\cite{sport}, climatology, \cite{stoker}, modeling of fluid flow
in homogeneous media, \cite{aziz}, catalysis, \cite{ellis},
combustion, \cite{glass}.

In a similar way nonlinear advection-diffusion-reaction equations
are obtained by using nonlinear conservation laws, see
\cite{tadmor}, \cite{bressan} and \cite{randal} . The nonlinear
advection-diffusion-reaction equation has the form
\begin{equation}
\frac{\partial{c(x,t)}}{\partial{t}}+\frac{\partial}{\partial
x}P\left(u(x,t)c(x,t)\right)=-\frac{\partial}{\partial
x}\left(J_c\right)+f(x,t,c(x,t)).
\end{equation}
where $P\left(u(x,t)c(x,t)\right)$ represents the non-linear convective flux.

\section*{Appendix B}

\vskip 0.2cm \noindent\textbf{\emph{Laplace transform of
fractional derivatives }}\vskip 0.2cm

\begin{align}\label{pre6}
\mathcal{L}[_0^*D_t^{\alpha}f(t);s]:&=s^{\alpha}\widetilde{f}(s)-s^{\alpha-1}f(0^+),
\quad 0 <\alpha\leq 1,\\
\nonumber \text{where} \\
\nonumber f(0^+):&=\lim_{t\rightarrow 0^+}f(t).\\
\mathcal{L}[_0D_t^{\alpha}f(t);s]:&=s^{\alpha}\widetilde{f}(s)-\left(_0I_t^{1-\alpha}f\right)(0^+),
\quad 0<\alpha\leq 1,\\
\nonumber \text{where} \\
\nonumber \left(_0I_t^{1-\alpha}f\right)(0^+):&=\lim_{t\rightarrow
0^+}\left(_0I_t^{1-\alpha}f\right)(t)\\
\mathcal{L}\left[D_t^{\alpha,\beta}f(t);s\right]:&=s^{\alpha}\widetilde{f}(s)-s^{\beta(1-\alpha)}\left[_0I^{(1-\beta)(1-\alpha)}f(0^+)\right],
\quad 0<\alpha<1,\\
\nonumber \text{where} \\
\nonumber
\left(_0I_t^{(1-\beta)(1-\alpha)}f\right)(0^+):&=\lim_{t\rightarrow
0^+}\left(_0I_t^{(1-\beta)(1-\alpha)}f\right)(t)
\end{align}

\vskip 0.2cm \noindent\textbf{Note: } \vskip 0.1cm \noindent One
can see that the differences in these Laplace transforms are in
the initial data $f(0^+)$, $\left(_0I_t^{1-\alpha}f\right)(0^+)$,
and $\left(_0I_t^{(1-\beta)(1-\alpha)}f\right)(0^+)$.

\begin{lemma}\label{lma3}\cite{furati2013non}
 Assume that $f(t)$
is continuous on $[0,T]$, for some $T>0$, then
$$\lim_{t\rightarrow
0^+}\left(_0I_t^{\alpha}f\right)(t)=0,$$ for $\alpha>0$.
\end{lemma}

%
%\vskip 0.2cm \noindent\textbf{Proof: } \vskip 0.1cm \noindent Let
%$0\leq t \leq T$, then
%\begin{align}%\label{pre7}
%\nonumber
%\left|\int_0^t(t-\tau)^{\alpha-1}f(\tau)d\tau\right|&\leq
%\int_0^t|t-\tau|^{\alpha-1}|f(\tau)|d\tau\\
%\nonumber &\leq M \int_0^t|t-\tau|^{\alpha-1}d\tau\\
%\nonumber &\leq M . \frac{-(t-\tau)^{\alpha}}{\alpha}|_0^t\\
%\nonumber &\leq \frac{M}{\alpha} t^{\alpha}\rightarrow 0 \text{ as
%} t\rightarrow 0.
%\end{align}
%


\begin{thebibliography}{0}


\bibitem{hund} W. Hundsdorfer, J.Verwer;
\emph{Numerical solution of time-dependent
advection-diffusion-reaction equations}, Springer-Verlag, Berlin,
Heidelberg, 2003.
%
\bibitem{bird} R.B. Bird, W.E. Stewart, E.N. Lightfoot;
\emph{Transport phenomena}, John Wiley and Sons, Inc., New York,
2002.
%
\bibitem{crut} T.E. Graedel, P.J. Crutzen;
\emph{Atmosphere, climate and change}, Henry Holt and Company,
1997.
%

\bibitem{sport}  B. Sportisse;
\emph{Air Pollution Modelling and Simulation}, Springer-Verlag,
Berlin, Heidelberg, 2002. %%

%
%
\bibitem{stoker} T. Stocker;
\emph{Introduction to Climate Modelling}, Springer, new York,
2011.

%
\bibitem{aziz} K. Aziz, A. Settari;
\emph{Petroleum reservoir simulation}, Applied Science Publishers
Ltd., London, 1979.
%

\bibitem{ellis} E. Cumberbatch, A. Fitt;
\emph{Mathematical Modeling: Case Studies from Industries},
Cambridge University Press, UK, 2001.
 %%



\bibitem{glass} I. Glassman, R.A. Yetter;
\emph{Combustion}, Academic Press, London, San Diego, Burlington,
2008. %%
%

\bibitem{tadmor} A. Kurganov, E. Tadmor;
\emph{New high-resolution central schemes for nonlinear
conservation laws and convection-diffusion equations}, J. Comp.
Physics, 160 (2000), 241-282.
%
\bibitem{bressan} A. Bressan, M. Lewicka, G. Chen, D. Wang;
\emph{Nonlinear conservation laws and applications}, Springer
Science, New York, 2011.
%
\bibitem{randal} R.J. LeVeque;
\emph{Nonlinear conservation laws and finite volume methods for
astrophysical fluid flow}, Computational methods for astrophysical
fluid flows, Springer-Verlag, Berlin, Heidelberg, 1998.
%

%
\bibitem{saenton} M. Khebchareon, S. Saenton;
\emph{Finite Element Solution for 1-D Groundwater Flow,
Advection-Dispersion and Interphase Mass Transfer : I. Model
Development}, Thai J.  Math., 3 (2005), 223-240. %%
%%
%%
\bibitem{marchuk}   G.I. Marchuk;
\emph{Mathematical Models in Environmental Problems},
North-Holland, Elsevier Science Publisher, 1986.
 %%

\bibitem{holton}   J.R. Holton, G.J. Hakim;
\emph{An introduction to dynamic meteorology},  Academic Press,
USA, 2013. %%
%
%
\bibitem{oran}  E. S. Oran, J. P. Boris;
\emph{Numerical Simulation of Reactive Flow},  Second edition,
Cambridge University Press, UK, 2001. %%
%
%%
\bibitem{glow}  R. Glowinski, J. Xu;
\emph{Numerical Methods for Non-Newtonian Fluids}, Handbook of
Numerical Analysis, Volume 16 (2011), 1-801. %%
%%
%%
\bibitem{menko}  F.V. Shuhaev, L.S. Shtemenko;
\emph{Propagation and reflection of shock waves}, World Scientific
Publishing, Singapore, 1998. %%
%%

%
\bibitem{kest} M. Treiber, A. Kesting;
\emph{Traffic flow dynamics: Data, Models and Simulation},
Springer-Verlag, Berlin, Heidelberg, 2013. %%
%%
\bibitem{piro} O. Pironneau, Y. Achdou;
\emph{Partial differential equations for option pricing}, Math.
Model. and Num. Methods in Finance, Special Volume, Handbook of
Numerical Analysis, 2008. %%
%%

\bibitem{jeon} J.H. Jeon, V. Tejedor, S. Burov, E. Barkai, C. Selhuber-Unkel, K. Berg-Sørensen,
L. Oddershede, and R. Metzler;
\emph{In vivo anomalous diffusion and weak ergodicity breaking of lipid granules},
Physical review letters 106, no. 4 (2011) 048103.
%%

\bibitem{saba} J. Sabatier, O. P. Agrawal, and J. A. T. Machado,
\emph{Advances in fractional calculus},
Dordrecht: Springer, 2007.
%%

\bibitem{tabei} S.M.A. Tabei, S. Burov, H. Y. Kim, A. Kuznetsov, T. Huynh, J. Jureller, L. H. Philipson,
A. R. Dinner, and N. F. Scherer;
\emph{Intracellular transport of insulin granules is a subordinated random walk},
Proceedings of the National Academy of Sciences 110, no. 13 (2013) 4911-4916.

%%
%
\bibitem{weiss} M. Weiss, M. Elsner, F. Kartberg, T. Nilsson;
\emph{Anomalous Subdiffusion Is a Measure for Cytoplasmic Crowding
in Living Cells}, Biophysical J., 87 (2004), 3518-3524.
%
%
%
\bibitem{chen} W. Chen, H. Sun, X. Zhang, D. Korosak ;
\emph{Anomalous diffusion modeling by fractal and fractional
derivatives}, Comp. Math. Appl. 59(5), (2010), 1754-1758.
%
\bibitem{old}   K.B. Oldham, J. Spanier;
\emph{The fractional calulus}, Academic Press, New York and
london, 1974.
%
\bibitem{mlros} K.S. Miller, B. Ross;
\emph{An Introduction to the fractional calculus and fractional
differential equations}, John Wiley and Sons, Inc., New York,
2003.
%
\bibitem{igor} I. Podlubny;
\emph{Fractional differential equations}, Academic Press, San
Diego, Calfornia, USA, 1999.
%
%
\bibitem{dieth} K. Diethelm, N.J. Ford;
\emph{Analysis of fractional differential equations}, J. Math.
Anal. Appl. 265 (2002) 229-248.
%

\bibitem{cap99} M. Caputo;
\emph{Diffusion of fluids in porous media with memory},
Geothermics, 28 (1999) 113-130.
%

%
\bibitem{havlin} S. Havlin, D. Ben-Avraham;
\emph{Diffusion in disordered media}, Advances in Physics,
51(2002), 187-292.
%

%
\bibitem{metzler1} R. Metzler, J. Klafter;
\emph{The random walk's guide to anomalous diffusion: a fractional
dynamics approach}, Physics Reports, 339 (2000) 1-77.
%
%
\bibitem{sdas} S. Das;
\emph{Analytical solution of a fractional diffusion equation by
variational iteration method}, Comp. Math. Appl. 57 (2009),
483-487.
%
\bibitem{saha}  S. Saha Ray, R.K. Bera;
\emph{Analytical solution of a fractional diffusion equation by
Adomian decomposition method}, Appl. Math. Comput. 174 (2006),
329-336.
%
\bibitem{hilfer1}  R. Hilfer;
\emph{Applications of fractional calculus in physics}, World
Scientific Publishing Company, Singapore, 2000.
%
\bibitem{hilfer2}  R. Hilfer;
\emph{Experimental evidence for fractional time evolution in glass
forming materials}, Chemical Physics, 284 (2002) 399-408.
%
\bibitem{hilfer3} R. Hilfer;
\emph{On fractional relaxation},
Fractals 11 (2003) 251-257.
%%

\bibitem{hilfer4} R. Hilfer,
\emph{Foundations of fractional dynamics: a short account},
Fractional dynamics: recent advances. World Scientific, Singapore 207 (2011).
%%

\bibitem{metzler2} T. Sandev, R. Metzler,  Z. Tomovski;
\emph{Fractional diffusion equation with a generalized
Riemann-Liouville time fractional derivative}, J. Phys. A: Math.
Theor. 44 (255203) (21pp), (2011).
%%

\bibitem{metzler3} R. Metzler, and J. Klafter,
\emph{The restaurant at the end of the random walk: recent developments in the
description of anomalous transport by fractional dynamics},
Journal of Physics A: Mathematical and General 37, no. 31 (2004): R161.

%%
\bibitem{tom} Ž.  Tomovski,  T. Sandev, R. Metzler, and J. Dubbeldam,
\emph{Generalized space–time fractional diffusion equation with composite fractional
time derivative}, Physica A: Statistical Mechanics and its Applications
391, no. 8 (2012): 2527-2542.

%%

\bibitem{kilbas} A.A. Kilbas, H.M. Srivastava, J.J. Trujillo;
\emph{Theory and Applications of Fractional differential
equations}, North-Holland, Elsevier Science Publisher, 2006.
%

\bibitem{srivastava2009fractional}H. M. Srivastava, Ž. Tomovski;
\emph{Fractional calculus with an integral operator containing a generalized Mittag–Leffler function in the kernel}'
Applied Mathematics and Computation 211.1 (2009): 198-210.

%
\bibitem{furati2013non} K. M. Furati,  M. D. Kassim, N. -e. Tatar;
\emph{Non-existence of global solutions for a differential equation involving Hilfer fractional derivative},
Electronic Journal of Differential Equations 2013.235 (2013): 1-10.


\bibitem{waz}  A.M. Wazwaz;
\emph{Partial differential equations and solitary waves theory},
Springer, New York, 2009.
%

\bibitem{zhao}  Y. Liu, X. Zhao;
\emph{He's Variational Iteration Method for Solving Convection
Diffusion Equations}, Advanced Intelligent Computing Theories and
Applications, Lecture Notes in Computer Science, Springer, New
York, 6215 (2010) 246-251.
 %

\bibitem{noor} Y. Molliq R., M.S.M. Noorani, I. Hashim;
\emph{Variational iteration method for fractional heat- and
wave-like equations}, Nonlinear Analysis: Real World Apllications,
10(2009) 1854-1869.
%

\bibitem{obidat} S. Momani, Z. Odibat;
\emph{Comparison between the homotopy perturbation method and the
variational iteration method for linear fractional partial
differential equations}, Comp. Math. App., 54(2007) 910-919.
%



\bibitem{he1} J.H. He;
\emph{Approximate analytical solution for seepage flow with
fractional derivatives in porous media}, Comput. Methods Appl.
Mech. Engrg. 167  (1998) 57-68.
%

%
\bibitem{he2}  J.H. He;
\emph{Variational iteration method: a kind of non-linear
analytical technique: some examples}, Int. J. Non-Linear Mech., 34
(1999) 699-708.
%


\bibitem{feng} F. Qi;
\emph{Bounds for the ratio of two Gamma functions}, J. Ineq. App.,
Volume 2010, Article ID 493058, 84 pages.
%
%
\bibitem{alzr} H. Alzer;
\emph{Sharp upper and lower bounds for the gamma function},
Proceedings of the Royal Society of Edinburgh: Section A
Mathematics, Appl. Numer. Math. 139 (2009), 709-718.
%%
\end{thebibliography}
\end{document}